\documentclass[journal,twocolumn]{IEEEtran}

\makeatletter
\def\endthebibliography{%
	\def\@noitemerr{\@latex@warning{Empty `thebibliography' environment}}%
	\endlist
}
\makeatother
\usepackage[T1]{fontenc}
\usepackage[latin9]{inputenc}

\usepackage{makecell,amsbsy,amsmath,amssymb,bbm,mathrsfs,multirow,amsthm,bm}
\usepackage{array,multirow,graphicx}
\usepackage[english]{babel}
\usepackage[justification=centering]{caption}
\usepackage{url}
\usepackage{threeparttable}
\usepackage{epstopdf}
\usepackage{subfigure}
\usepackage{color,xcolor}

\usepackage{algorithm,algpseudocode}

\DeclareMathAlphabet{\mathpzc}{OT1}{pzc}{m}{it}

\usepackage{url}

\def\cF{\mathcal F}
\def\cG{\mathcal G}
\def\cI{\mathcal I}
\def\cH{\mathcal H}

\def\cK{\mathcal K}

\def\cX{\mathcal X}
\def\cQ{\mathcal Q}

\def\cP{\mathcal P}
\def\cR{\mathcal R}

\newcommand{\ba}{\text{\bf a}}

\newcommand{\bp}{\text{\bf p}}

\newcommand{\by}{\text{\bf y}}

\newcommand{\bs}{\text{\bf s}}
\newcommand{\bu}{\text{\bf u}}
\newcommand{\bv}{\text{\bf v}}

\newcommand{\trmp}{\ensuremath{\mathsf{TRMP}}}
\newcommand{\ub}{\ensuremath{\mathsf{UB}}}

\usepackage[font=footnotesize]{caption}

\makeatletter
\renewcommand{\fnum@figure}{Fig. \thefigure}
\makeatother

\newtheorem{lemma}{Lemma}
\newtheorem{theorem}{\textbf{\textsc{Theorem}}}

\IEEEoverridecommandlockouts

\begin{document}
	\bibliographystyle{IEEE2}
	
	\title{Joint Rate Allocation and Power Control for RSMA-Based Communication and Radar Coexistence Systems}
	\author{Trung Thanh Nguyen, Nguyen Cong Luong, Shaohan Feng, Tien Hoa Nguyen, Khaled Elbassioni, Dusit Niyato, ~\IEEEmembership{Fellow,~IEEE}, and
		Dong In Kim,~\IEEEmembership{Fellow,~IEEE}

		\thanks{T. T. Nguyen and N. C. Luong are with Faculty of Computer Science, Phenikaa University, Hanoi 12116, Vietnam (e-mails: \{thanh.nguyentrung, luong.nguyencong\}@phenikaa-uni.edu.vn.)}

		\thanks{S. Feng is with Institute for Infocomm Research, Singapore 138632 (e-mail: fengs@i2r.a-star.edu.sg).}
		
		\thanks{T. H. Nguyen is with School of Electrical and Electronic Engineering, Hanoi University of Science and Technology, Hanoi 100000, Vietnam (e-mail: hoa.nguyentien@hust.edu.vn).}
		
		\thanks{K. Elbassioni is with Khalifa University of Science and Technology, Abu Dhabi, United Arab Emirates (e-mail: khaled.elbassioni@ku.ac.ae).} 
		
		\thanks{D. Niyato is with the School of Computer Science and Engineering, Nanyang Technological University, Singapore (e-mail: dniyato@ntu.edu.sg).}
		
		\thanks{D.~I.~Kim is with the Department of Electrical and Computer Engineering, Sungkyunkwan University, Suwon 16419, South Korea (e-mail:dikim@skku.ac.kr).}
	}
	
	
	\maketitle

	\begin{abstract}
		We consider a rate-splitting multiple access (RSMA)-based communication and radar coexistence (CRC) system. The proposed system allows an RSMA-based communication system to share spectrum with multiple radars. Furthermore, RSMA enables flexible and powerful interference management by splitting messages into common parts and private parts to partially decode interference and partially treat interference as noise. The RSMA-based CRC system thus significantly improves spectral efficiency, energy efficiency and quality of service (QoS) of communication users (CUs). However, the RSMA-based CRC system raises new challenges. Due to the spectrum sharing, the communication network and the radars cause interference to each other, which reduces the signal-to-interference-plus-noise ratio (SINR) of the radars as well as the data rate of the CUs in the communication network. Therefore, a major problem is to maximize the sum rate of the CUs while guaranteeing their QoS requirements of data transmissions and the SINR requirements of multiple radars. To achieve these objectives, we formulate a problem that optimizes i) the common rate allocation to the CUs, transmit power of common message and transmit power of private messages of the CUs, and ii) transmit power of the radars. The problem is non-convex with multiple decision parameters, which is challenging to be solved. We propose two algorithms. The first  sequential quadratic programming (SQP) can quickly return a local optimal solution, and has been known to be the state-of-the-art in nonlinear programming methods. The second is an additive approximation scheme (AAS) which solves the problem globally in a reasonable amount of time, based on the technique of applying exhaustive enumeration to a modified instance. The simulation results show the improvement of the AAS compared with the SQP in terms of sum rate. Furthermore, with the AAS, the sum rate of the CUs only slightly decreases when the radars' SINR is significantly increased. This implies that the AAS supports the RSMA-based communication system which allows to well coexist with the radars.
	\end{abstract}
	
	
	\section{Introduction}
	With the rapid growth of multimedia applications, e.g., virtual reality, and mobile devices, spectrum resource is becoming increasingly congested. As a consequence, mobile network operators seek to reuse spectrum of other systems~\cite{liu2020joint} as well as advanced wireless technologies. In particular, coexisting radar and communication (CRC) systems have been recently considered that allow an individual communication system to share spectrum with radar systems~\cite{luong2021radio}. Such a CRC system can be found in many realistic scenarios. For example, the air-surveillance radars and the 5G NR and FDD-LTE cellular systems share the L-band, i.e., $1-2$ GHz~\cite{wang2017spectrum}. As a different example, the early warning radars share the S-band, i.e., $2-4$ GHz, with the communication systems like 802.11b/g/n/ax/y WLAN networks and $3.5$ GHz TDD-LTE and 5G NR~\cite{hessar2016spectrum}. Meanwhile, for advanced wireless technologies, rate-splitting multiple access (RSMA)~\cite{dizdar2020rate}, \cite{mao2022rate}, \cite{clerckx2016rate} has emerged as a promising solution for the next generation network that is able to achieve high spectrum efficiency, robust, and high data rate. With RSMA, a base station (BS) as a transmitter splits each message intended for a user into a common part and a private part. The common parts of all users are combined into a common massage and can be decoded by all the users, while the private part is encoded into a private message which can be decoded by only its intended user. To receive the common message, the user suffers the interference from all the users' private messages, and to receive the private message, the user only considers the interference from other users' private messages to be noise. Each user then recovers its original message from its common message and intended private message with successive interference cancellation (SIC). By splitting messages into common parts and private parts to partially
	decode interference and partially treat interference as noise, RSMA enables flexible and powerful interference management, which enhance spectral
	efficiency, energy efficiency,
	reliability, and quality of service (QoS) compared with existing multiple access technologies such as space division multiple access (SDMA) and non-orthogonal multiple access (NOMA)~\cite{dizdar2020rate}, \cite{mao2018energy}. 
	
	The benefits of CRC and RSMA motivate us to investigate a novel system, namely RSMA-based CRC system. The system allows a communication network to leverage RSMA to serve multiple communication users (CUs) and share spectrum resources with multiple radars. The RSMA-based CRC system is thus expected to significantly improve the spectrum efficiency. However, the coexistence of the communication network and multiple radars in the RSMA-based CRC system imposes challenges of resource management. In particular, the major problem is to properly adjust the common rates, transmit power of the common message, and private messages for the CUs as well as to control the transmit power of the radars to maximize the sum rate of all the CUs, subject to the QoS requirements of the CUs and the signal-to-interference-plus-noise ratio (SINR) requirements of the radars. Like most of the optimization problems in communication system design, this problem belongs to the class of fractional programming problems, which involves ratios of linear functions and is known to be computationally NP-hard \cite{FreundJ01}.  Furthermore,  due to the minimum requirement for the rate of CUs, checking if there is a feasible solution is as difficult as the original problem itself. Two methods have been well studied  
	for finding locally optimal solution to fractional programming problems in the context of communication system such as Successive Convex Approximation \cite{WangV12,ScutariFSPP14,YeC20,yang2021optimization} and iterative algorithms based on Dinkelbach's transformation (see \cite{ShenY18} and the references therein), but they both require an initial feasible solution as input and thus are inapplicable to our problem. On the other hand, existing global optimization algorithms have been mainly based on the branch-and-bound technique \cite{lawler1966branch}, and only applied to linear constraints. The cases of nonlinear constraints are much more complicated and may require the development of more general and sophisticated techniques.  
	
	To the best of the author's knowledge, this is the first work investigating the coexistence of the RSMA-based communication network and multiple radars. Indeed, there exist recent works, i.e., \cite{shi2018non}, \cite{shi2019robust}, \cite{li2017joint}, \cite{liu2018transmission}, \cite{liu2017robust}, \cite{qian2022transmission}, \cite{hong2021interference}, \cite{rihan2018optimum}, \cite{hong2019ergodic}, and \cite{kafafy2022optimal} related to CRC, but non of them considers the use of RSMA for the communication networks. In particular, the work in \cite{shi2018non} considers a traditional communication system coexisting with multiple radars. The work aims to minimize the power consumption of each radar by optimizing the transmission power allocation, subject to the radars' SINR requirements and a maximum interference tolerant limit for communication network. The authors in \cite{li2017joint} consider a CRC system with a single communication user and a radar. The objective is to maximize the radar's SINR, subject to the communication rate and power constraints. The works in \cite{liu2018transmission} and \cite{liu2017robust} extend the communication system in~\cite{li2017joint} to scenarios with multiple downlink users. The work in \cite{qian2022transmission} aims to optimize the waveforms of the radar system and the codebook of the communication system, subject to the radar similarity with a radar reference waveform and the data rate requirement of the communication system. Different from~\cite{qian2022transmission}, the work in \cite{rihan2018optimum}  considers a collocated radar and communication system, and thus the joint transmit and receive beamformers for both the radar and communication systems is optimized to maximize SINR of both the radar and communication systems. Unlike \cite{rihan2018optimum}, the work in \cite{hong2021interference} leverages the interference existing already in radar and communications systems to keep constructive and can contribute to the power of the useful signal. Simulation results show that the signal-to-noise ratio (SNR) can be enhanced by $7$ dB. The channels involved in the CRC systems can vary over time that can reduce the effectiveness of the interference management approaches. Consider this issue, the authors in \cite{hong2019ergodic} propose an ergodic interference alignment, which allows to time-varying and/or frequency selective channels for the CRC system. In particular, this work designs a transmit beamforming, which helps to find a pair of desired complementary channels easily without long-time system delay, thus effectively eliminating the interference in the CRC system. Another traditional solution for the interference elimination between the radar system and communication system is to place the communication system outside a guard zone of the radar. However, the CUs in the guard zone have bad QoS experience since they are out of the coverage of the communication system. To address this issue, the work in 	\cite{kafafy2022optimal} proposes to deploy intelligent reflecting surface (IRS) in the guard zone to enhance the network coverage and QoS of the users.
	
	There also exist several works related to RSMA, and the readers are referred to a comprehensive survey on RSMA in~\cite{mao2022rate}. However, the existing works do not investigate CRC. In particular, the authors in~\cite{clerckx2019rate} and \cite{mao2018energy} evaluate the performance obtained by RSMA. These works demonstrate that RSMA outperforms both NOMA and SDMA in terms of spectrum efficiency, energy efficiency, QoS, and computational
	complexity in different network loads and channel conditions. The work in \cite{joudeh2016sum} aims to optimize the beamformers associated with the common message and private messages to maximize the sum rate of CUs in  multi-user multiple input single output (MISO) systems under imperfect
	channel state information. The authors in \cite{hao2015rate} analyze the data rate obtained by
	RSMA for two-receiver MISO broadcast channel with finite rate feedback. Considering a multi-user single input single output (SISO), the authors in~\cite{yang2021optimization} aim to allocate the common rate, transmit power of the common message, and transmit power of the private messages by accounting for successful SIC power requirements. Different from \cite{joudeh2016sum}, \cite{hao2015rate}, and \cite{yang2021optimization}, the works in~\cite{joudeh2016robust} and \cite{yin2021rate} aim to achieve max-min fairness among the users, i.e., maximizing the user capacity with the minimum data rate, in the RSMA-based network. Recently, RSMA can also be combined with enabling technologies in 6G. In particular, RSMA is shown to efficiently manage the inter-beam interference for satellite communications ~\cite{yin2020rate}. RSMA can be combined with simultaneous wireless information and power transfer (SWIPT). As presented in~\cite{mao2019rate}, RSMA-based SWIPT can significantly improve the sum rate of information receivers (IRs) compared with NOMA-based SWIPT. The combination of RSMA and intelligent reflecting surface (IRS) has emerged as a promising solution that can improve the energy efficiency~\cite{yang2020energy} and reduce the outage probability~\cite{bansal2021rate}. Some recent works, e.g., \cite{dizdar2022energy} and ~\cite{xu2021rate}, investigate an integration of RSMA into a radar system on the same hardware, namely RSMA-enabled dual-function radar communication (DFRC). In the RSMA-enabled DFRC system, an RSMA signal is designed to simultaneously perform both the communication function and radar function. As such, apart from achieving the high spectrum efficiency, the RSMA-enabled DFRC system reduces the hardware size and cost. However, sharing the hardware can impair the performance of each function. Moreover, it is costly or even impossible to deploy such a system in existing radar system or communication network.
	
	In summary, the works related to CRC have not considered RSMA as the communication system and leverage the effective interference management of RSMA to handle the interference from the radars. On the other hand, in the works related to RSMA, none of them has considered that the RSMA can effectively manage the co-channel interference between the communication and radar systems so as to enhance their performance simultaneously in the CRC scenario. To bridge this gap, we investigate the RSMA-based CRC system. The main contributions of this work are follows:

	\begin{itemize}
		\item We consider a novel communication and radar coexistence system, namely RSMA-based CRC. In the RSMA-based CRC system, a BS in the communication network shares spectrum with multiple radars while leveraging the RSMA scheme to serve its multiple CUs. The RSMA-based CRC system thus achieves the advantages of both RSMA and CRC, which particularly enhances the spectrum efficiency. Due to the spectrum sharing, the radar and communication systems cause interference to each other. Thus, the objective is to maximize the sum rate of all the CUs in the communication network subject to the requirements of the CUs' QoS and the radars' SINR as well as the power budgets of the BS and the radars. 
		\item To achieve the objective, we formulate a problem that optimizes i) common rates of the CUs, ii) transmit power of the common message of the CUs, iii) transmit power of the private messages for the intended CUs, and iv) transmit power of the radar systems. 
		\item The optimization problem is non-convex, which is challenging to solve. To solve it, we first present a sequential quadratic programming (SQP) algorithm~\cite{boggs1995sequential} which iteratively converges to a local optimal solution. 
		\item As a major contribution, we propose a more general technique for solving the problem globally. Particularly, we design an additive approximation scheme  AAS\footnote{AAS stands for Additive Approximation Scheme, which is a family of algorithms that, for each $\delta>0$, return a solution with an absolute error in the objective of at most $\delta\cdot h$ for some suitable parameter $h$ \cite{V0004338}.}, based on the technique of applying exhaustive enumeration to a modified instance. The algorithm runs in exponential time in the number of users, and produces a solution whose value is within an additive error controlled by a parameter $\delta\in (0,1)$. 
		\item We provide the simulation results to evaluate the proposed algorithms. The simulation results insightfully show that the AAS significantly improves the performance in terms of sum rate compared with the SQP. The simulation results further show that the combination of the AAS and RSMA is an effective solution that allows the communication system to well coexist with the radars. In addition, appropriate locations of the radars can be suggested. 
	\end{itemize}

	For the reader's convenience, we summarize the mathematical notations used in this paper in Table~\ref{table:major_notation}. Some other symbols are given in Table~\ref{table:parameters}. 
	\begin{table}[h!]
		\caption{List of mathematical symbols used in this paper.}
		\label{table:major_notation}
		\centering
		\begin{tabular}{|l|l|}
			\hline
			\textbf{Notation} & \textbf{Description}   \\ [1ex]
			\hline $|V|$ &The size of a set $V$\\
			\hline $\cQ$ &The set of CUs\\
            \hline $\cK$ &The set of radars\\
            \hline $Q=|\cQ|$ &The number of CUs\\
            \hline $K=|\cK|$ &The number of radars\\
			\hline $B$ &The total bandwidth\\
			\hline $\sigma^2_q$&  The noise variance at CU $q$\\
			\hline $\sigma^2_k$&  The noise variance at radar $k$\\
			\hline $n^{\rm{C}}_{q}\sim \mathcal{N}(0,\sigma_{q}^2)$& The noise at the CU $q$\\
			\hline $n^{\rm{R}}_{k}\sim \mathcal{N}(0,\sigma_{k}^2)$& The noise at the radar $k$\\
			\hline $p_0$&  The transmit power of common message $s_0$\\
			\hline $p_q$&  The transmit power of private message $s_q$\\
			\hline $p_k^{\rm{R}}$ &The transmit power of radar $k$\\
			\hline $h^{\rm{C}}_q$ &The channel gain between the BS and CU $q$\\  
			\hline $g^{\rm{RC}}_{k,q}$ & The channel gain between radar $k$ and CU $q$\\
			\hline $h^{\rm{R}}_{k}$ & The round-trip channel gain of radar $k$\\ 
			\hline $g^{\rm{RR}}_{k',k}$ & The channel gain of the direct link from the TX\\
			 &  of radar $k'$ to the RX of radar $k$\\
			\hline $g^{\rm{RTR}}_{k',k}$ & The channel gain of the indirect link of the TX\\
			 & of radar $k'$, the target and the RX of radar $k$\\ 
			\hline $h^{\rm{CR}}_k$ & The channel gain between the BS and radar $k$\\
			\hline $f_i,g_i$ &The radar and communication functions of user $i$\\
            \hline $\bp$        & The power vector of dimension ${Q+K+1}$ \\
            \hline $\ba$        & The common data rate vector of dimension $Q$ \\
			\hline $\cF(\ba,\bp)$  &The multivariate function of variables $(\ba,\bp)$\\
			\hline $L$  &The Lagrange function\\
			\hline $\nabla\cF$  & The Gradient of function $\cF$\\
			\hline $\nabla^2\cF$  & The Hessian matrix of function $\cF$\\
			\hline$\epsilon,\delta$& The additive errors of algorithms\\
			\hline
		\end{tabular}
		\label{table:parameters}
	\end{table}

	The rest of the paper is organized as follows. In Section~\ref{sys_model}, we present the RSMA-based CRC system and formulate the optimization problem. Sections~\ref{sec:iter_algorithm} and~\ref{sec:exact_algorithm} are devoted to presenting the two proposed algorithms, SQP and AAS, respectively. Section~\ref{sec:perform-eval} provides simulation results to demonstrate the effectiveness of the proposed algorithm. The conclusions of this paper are given in Section~\ref{conclusion}.

	\section{System Model}
	\label{sys_model}
	
	In this section a novel communication and radar coexistence system, namely the RSMA-based CRC system is introduced with details. We then describe the problem of maximizing the sum rate of the CUs in the system under various constraints, and give it a mathematical formulation.
	\subsection{Signal Model}
	
	\begin{figure*}[h!]
		\centering
		\includegraphics[width=0.9\textwidth]{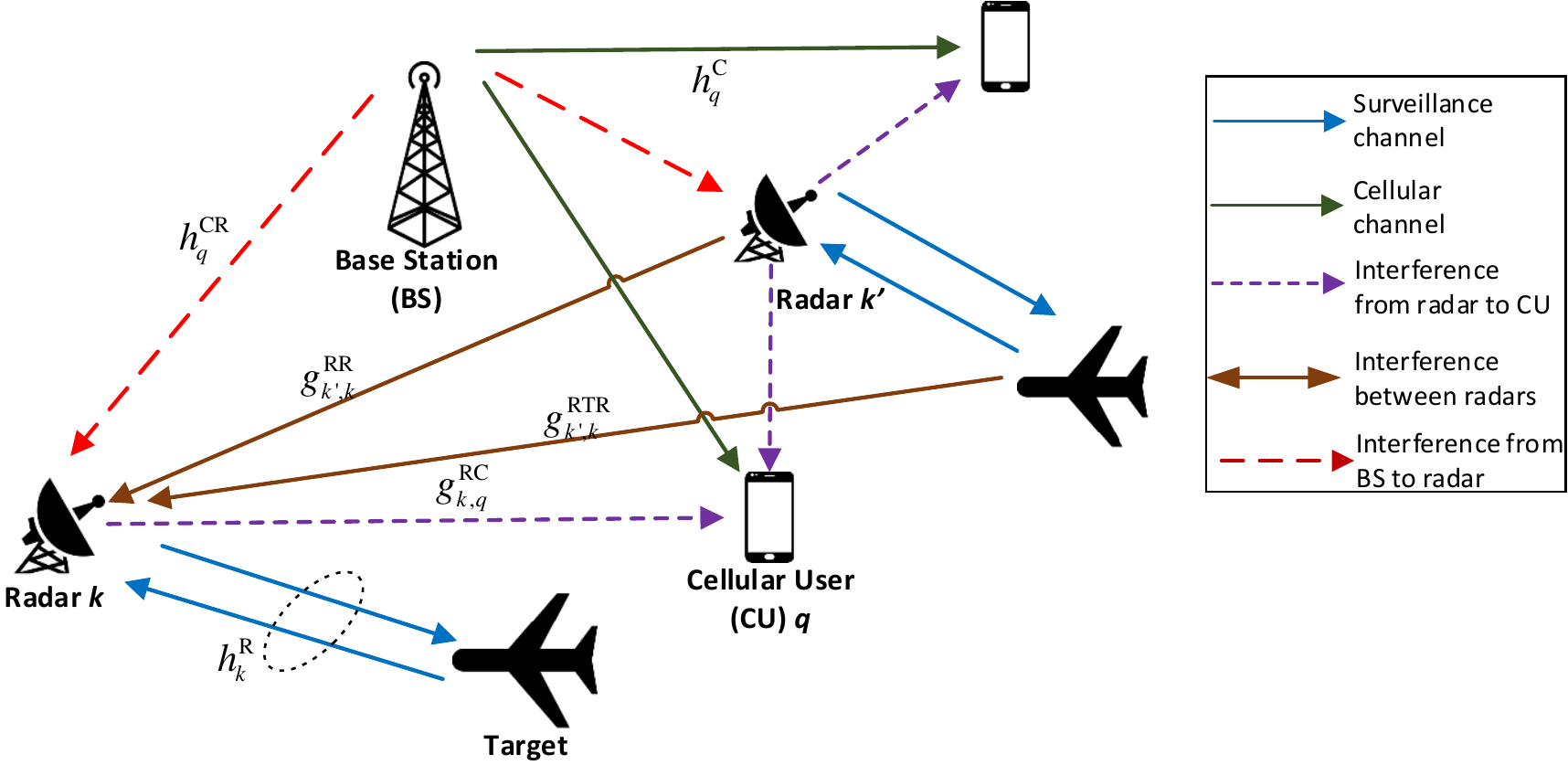}
		\caption{Illustration of the RSMA-based communication and radar coexistence (RSMA-based CRC) system.}
		\label{fig:system_model}
	\end{figure*}
	
	The RSMA-based CRC is shown in Fig.~\ref{fig:system_model}. Each radar is monostatic in which its transmitter (denoted by ``TX'')and receiver (denoted by ``RX'') are co-located. The radars perform tracking a common target or multiple targets. The set of the radars is denoted by ${\cal{K}}$ and the set of CUs is defined as ${\cal{Q}}$. The cardinalities of ${\cal{K}}$ and ${\cal{Q}}$ are $K$ and $Q$, respectively. By using RSMA, the BS transmits a common message of all the users and private messages to the intended users. In particular, the transmit signal from the BS is given by~\cite{yang2021optimization}
	\begin{align*}
		x = \sqrt {{p_0}} {s_0} + {\sum}_{q = 1}^Q {\sqrt {{p_q}} {s_q}},
	\end{align*}
	where $s_0$ is the common message, $s_q$ is the private message intended for the $q$-th user, $p_0$ is the transmit power of the common message $s_0$, and $p_q$ is the transmit power of private message $s_q$. The received signal at CU $q$ consists of the signal from the BS and the radar signals directly emitted from $K$ radars, which is given by~\cite{shi2018non}
	\begin{equation*}\label{eq:received_comm_signal}
		y^{\rm{C}}_{q} = \sqrt {{p_0}h^{\rm{C}}_q} {s_0} + {\sum}_{q = 1}^Q {\sqrt {{p_q}h^{\rm{C}}_q} {s_q}} + {\sum}_{k=1}^{K}  \sqrt{p^{\rm{R}}_{k}g^{\rm{RC}}_{k,q}}s^R_k + {n^{\rm{C}}_{q}},
	\end{equation*}
	where $h^{\rm{C}}_q$ is the channel gain between the BS and CU $q$, $g^{\rm{RC}}_{k,q}$ is the channel gain between radar $k$ and CU $q$, $s^R_k$ is the radar signal, $p_k^{\rm{R}}$ is the transmit power of radar $k$, and ${n^{\rm{C}}_{q}}\sim \mathcal{N}(0,\sigma_{q}^2)$ is the AWGN noise at the CU with $\sigma_{q}^2$ being the variance. 
	
	The received signal at RX of radar $k$ consists of i) the radar signals reflected by the target, ii) the radar signals directly emitted by other radars, and iii) the RSMA signal from the BS. Thus, the received signal at RX of radar $k$ is expressed by
	\begin{multline*}\label{eq:received_radar_signal}
		y^{\rm{R}}_{k} =  \sqrt{p^{\rm{R}}_kh^{\rm{R}}_k} s^R_k
		+ {\sum}_{k'=1,k'\neq k}^K  \sqrt{p^{\rm{R}}_{k'}{ (g^{\rm{RR}}_{k',k} + g^{\rm{RTR}}_{k',k}) }}s^R_{k'} \\
		+ \sqrt {{p_0}h^{\rm{CR}}_k} {s_0} + {\sum}_{q = 1}^Q {\sqrt {{p_q}h^{\rm{CR}}_k}} {s_q} + {n^{\rm{R}}_{k}},
	\end{multline*}
	where $h^{\rm{R}}_{k}$ is the round-trip channel gain of radar $k$, i.e., the channel gain of the indirect link of the TX-target-RX of radar $k$, $g^{\rm{RR}}_{k',k}$ is the channel gain of the direct link from the TX of radar $k'$ to the RX of radar $k$, $g^{\rm{RTR}}_{k',k}$ is the channel gain of the indirect link of the TX of radar $k'$, the target and the RX of radar $k$, $h^{\rm{CR}}_k$ is the channel gain between the BS and radar $k$, and ${n^{\rm{R}}_{k}}\sim \mathcal{N}(0,\sigma^2_k)$ is the noise at the radar.
	
	The channels involved in the RSMA-based CRC system are assumed to be quasi-static fading channel models that are defined as follows~\cite{shi2018non}, \cite{zhang2021design}:
	
	\begin{subequations} 
		\begin{align}
			\nonumber
			{}& \begin{cases} 
				\vspace{4pt}
				h^{\rm{R}}_{k}=\frac{G_{k,t}^{\rm{R}} G_{k,r}^{\rm{R}}\sigma_k^{\rm{RCS}}\lambda_c^2}{(4\pi)^3 (d_k^{\rm{R}})^4} \\
				\vspace{4pt}
				h^{\rm{CR}}_{k}=\frac{G_{t}^{\rm{C}} G_{k,r}^{'\rm{R}}\lambda_c^2  }{(4\pi)^2 (d_{k}^{\rm{CR}})^2} \hat{h}^{\rm{CR}}_{k} \\
				\vspace{4pt}
				h^{\rm{C}}_{q}=\frac{G_{t}^{\rm{C}} G_{q}\lambda_c^2}{(4\pi)^2 (d_{q})^2}\hat{h}^{\rm{C}}_{q}\\
				\vspace{4pt}
				g^{\rm{RTR}}_{k',k}=\frac{G_{k,t}^{\rm{R}} G_{k',r}^{\rm{R}}\sigma_{k,k'}^{\rm{RCS}}\lambda_c^2}{(4\pi)^3 (d_k^{\rm{R}})^2(d_{k'}^{\rm{R}})^2}\\
				\vspace{4pt}
				g^{\rm{RR}}_{k',k}=\frac{G_{k',t}^{'\rm{R}} G_{k,r}^{'\rm{R}}\lambda_c^2}{(4\pi)^2 (d_{k',k}^{\rm{RR}})^2}\hat{g}^{\rm{RR}}_{k',k}
				\\
				\vspace{4pt}
				g^{\rm{RC}}_{k,q}=\frac{G_{k,t}^{'\rm{R}} G_{q}\lambda_c^2}{(4\pi)^2 (d_{k,q}^{\rm{RQ}})^2}\hat{g}^{\rm{RC}}_{k,q} 
			\end{cases} 
		\end{align}
	\end{subequations}
	where $\hat{h}^{\rm{CR}}_{k}$, $\hat{h}^{\rm{C}}_{q}$, $\hat{g}^{\rm{RR}}_{k,'k}$, and $\hat{g}^{\rm{RC}}_{k,q}$ are the small-scale channels between the BS and radar $k$, the BS and CU $q$, radar $k'$ and radar $k$, and radar $k$ and CU $q$, respectively, $G_{k,t}^{\rm{R}}$ and $G_{k,r}^{\rm{R}}$ is the main-lobe antenna gains of the TX and RX of radar $k$, respectively, $G_{k,t}^{'\rm{R}}$ and $G_{k,r}^{'\rm{R}}$ are the side-lobe antenna gain of the TX and RX of radar $k$, respectively, $G_{t}^{\rm{C}}$ is the transmitting antenna gain of the BS, $G_{q}$ is the receiving antenna gain of CU $q$,  $\sigma_k^{\rm{RCS}}$ is the radar cross section (RCS) of the target with respect to radar $k$, $\sigma_{k,k'}^{\rm{RCS}}$ is the RCS of the target from radar $k$ to radar $k'$, $d_k^{\rm{R}}$, $d_{k',k}^{\rm{RR}}$, $d_{k,q}^{\rm{RQ}}$, $d_q$, $d_{k}^{\rm{CR}}$ are the distances from radar $k$ to its tracking target, radar $k$ to radar $k'$, radar $k$ to CU $q$, BS to CU $q$, BS to radar $k$, respectively. 
	
	\subsection{Optimization Problem}
	We first determine the data rate achieved by each CU. With the RSMA, the data rate achieved by each CU is the sum of common data rate and private data rate. Given a power vector $\bp=(p_0,p_1,\ldots,p_Q,p^{\rm{R}}_1,\ldots,p^{\rm{R}}_K)$, the common data rate and the private data rate of user CU $q$ are defined as functions of $\bp$ as follows~\footnote{For ease of notation, we use the binary logarithm in $\log(\cdot)$.}
	\begin{equation*}\label{eq:rate_comm-1}
		R^{\rm{C}}_{q,0}(\bp)= B\log \left( 1 + \frac{h^{\rm{C}}_q p_0}{h^{\rm{C}}_q {\sum}_{q' = 1}^Q p_{q'} +   {\sum}_{k=1}^{K} g^{\rm{RC}}_{k,q} p^{\rm{R}}_k + \sigma^2_q} \right),
	\end{equation*}
	and 
	\begin{equation*}\label{eq:rate_private-1}
		R^{\rm{C}}_{q}(\bp)= B\log  \left(  1 + \frac{h^{\rm{C}}_qp_q}{ h^{\rm{C}}_q{\sum}_{{q' = 1, q'\neq q}}^Q p_{q'} +   {\sum}_{k=1}^{K} g^{\rm{RC}}_{k,q}p^{\rm{R}}_k + \sigma^2_q} \right).
	\end{equation*}
	
	Denote $a_{q}$ as the common data rate allocated to CU $q$. Then, we have the following constraint
	
	\begin{equation}
		{\sum}_{q=1}^Qa_{q} \leq \min \limits_{q\in\cQ} \left\{R^{\rm{C}}_{q,0}(\bp)\right\}.
	\end{equation}
	The total data rate, denoted by $C^{\rm{C}}_q$, achieved by CU $q$ is given by
	\begin{equation}
		C^{\rm{C}}_q= a_{q}+ R^{\rm{C}}_{q}(\bp).
	\end{equation}
	
	As the communication system shares spectrum with the radar systems, there is a constraint on the interference caused by the radar systems to the CUs, which can be expressed by the data rate requirements of CUs as follows:
	\begin{equation}
		C^{\rm{C}}_q =a_{q}+ R^{\rm{C}}_{q}(\bp)\ge C^{\rm{TH}}_q, \quad \forall q\in \mathcal{Q},
	\end{equation}where $C^{\rm{TH}}_q$ is the minimum rate requirement of  CU $q$. It is worth noting that, by taking these constraints into account, a
		specific level of the desired quality of fairness could be achieved, while maximizing the total data rate of users.
	
	We denote $\vartheta^{\rm{R}}_{k}(\bp)$ as the SINR at the RX of radar $k$. Then, $\vartheta^{\rm{R}}_{k}(\bp)$ is a function of the power vector $\bp$, and defined by~\cite{shi2018non}, \cite{zhang2021design}
	\[
	\frac{ {h^{\rm{R}}_k} p^{\rm{R}}_k}
	{\sum_{{k'=1,k'\neq k}}^{K} \left(g^{\rm{RR}}_{k',k} + c_{k',k} g^{\rm{RTR}}_{k',k} \right)p^{\rm{R}}_{k'} +  h^{\rm{CR}}_k \left( \sum_{q = 0}^Q {{p_q}} \right)  + \sigma_k^2}.
	\]
	To guarantee the tracking performance, the SINR at each radar must be larger than a threshold, which is denoted by $\gamma^R$. Thus, we have the condition for the radars as $\vartheta^{\rm{R}}_{k}(\bp) \geq \gamma^R$.
	
	We then formulate an optimization problem, namely transmission rate maximization problem ($\trmp$) for the communication system, subject to the data rate requirements of the CUs and the SINR requirements of the radars as follows:		
	\begin{subequations}\label{prob:main}
		\begin{align}
			\max \limits_{\ba,\bp } & \quad {\sum}_{q=1}^Q a_{q} +{\sum}_{q=1}^Q R^{\rm{C}}_q(\text{\bf p}) \tag{\ref{prob:main}}\\
			{\text{s.t.}} & \quad 	{\sum}_{q=1}^Qa_{q} \leq \min \limits_{q\in\cQ} \left\{R^{\rm{C}}_{q,0}(\bp)\right\},\label{eq:min-R-q-0}\\
			&\quad 	a_{q}+ R^{\rm{C}}_{q}(\bp)\ge C^{\rm{TH}}_q, \hspace{2cm} \forall q\in \mathcal{Q},\label{eq:min-rate-user}\\
			&\quad p_0+ {\sum}_{q=1}^Q p_q \le {\bar{p}}^{\rm{C}},\label{eq:total-power}\\
			&\quad \vartheta_k(\bp)\ge \gamma^{\rm{R}}, \hspace{3.2cm} \forall k \in {\cal{K}},\label{eq:SINR}\\
			&\quad p_{k}^{\rm{R}} \le {\bar{p}}^{\rm{R}}, \hspace{3.7cm} \forall k \in {\cal{K}},\label{eq:radar-power}\\
			&\quad a_{q},p_0,p_q,p_k^{\rm{R}}\ge 0, \hspace{1.2cm} \forall q\in \mathcal{Q},\forall k \in {\cal{K}},
		\end{align}
	\end{subequations}
	where $\text{\bf a} = (a_{q})_{q\in \mathcal{Q}}$ and $\bp=(p_0,\bp^{\rm{U}},\bp^{\rm{R}})$ with $\bp^{\rm{U}}=(p_q)_{q\in \mathcal{Q}}$,  $\bp^{\rm{R}}=(p_k^{\rm{R}})_{k\in\mathcal{K}}$, $\gamma^{\rm{R}}$ is the minimum performance requirement for the radars, ${\bar{p}}^{\rm{R}}$ is the power budget at the radars, and ${\bar{p}}^{\rm{C}}$ is the power budget of the amplifiers of the BS.
	
	Note that by denoting
	\[
	\tilde{g}_{k',k}=\frac{g^{\rm{RR}}_{k',k}}{h^{\rm{R}}_k} + c_{k',k} \frac{g^{\rm{RTR}}_{k',k}}{h^{\rm{R}}_k},\,\tilde{h}_{k}=\frac{h^{\rm{CR}}_{k}}{h^{\rm{R}}_k},\, \tilde\sigma_k =\frac{\sigma^2_k}{h^{\rm{R}}_k},
	\]
	the constraint in (\ref{eq:SINR}) can be written in the linear form as:
	\begin{equation*}
		\frac{p^{\rm{R}}_k}{\gamma^{\rm R}}    -{\sum}_{k'=1,k'\neq k}^{K} \tilde{g}_{k',k}p^{\rm{R}}_{k'} -\tilde h_k {\sum}_{q = 0}^Q {{p_q}}  \ge   \tilde\sigma_k,\quad \forall\, k\in\cK.
	\end{equation*}
	To simplify the presentation of the problem model later on, let $\cX$ denotes the set of all points $\bp$ that are feasible to the linear constraints of~(\ref{prob:main}), and we can assume\footnote{Checking whether $\cX$ is empty is equivalent to checking the feasibility of a system of linear function, which can be done in polynomial time.} without loss of generality that $\cX \neq \emptyset$.
	
	In general, the transmission rate maximization problem $\trmp$ given in (\ref{prob:main}) is nonconvex due to the nonconcavity of the objective function and the nonconvexity of the constraints (\ref{eq:min-R-q-0}) and (\ref{eq:min-rate-user}), which make the problem hard to solve. To solve the problem, we propose two algorithms, i.e., an iterative SQP algorithm and an additive approximation scheme (AAS), which are presented in Sections~\ref{sec:iter_algorithm} and~\ref{sec:exact_algorithm}, respectively.
	
	\section{A Local Optimization Algorithm}\label{sec:iter_algorithm}
	In section, we present the use of the SQP algorithm~\cite{boggs1995sequential}, which is one of the most powerful algorithmic tools for the numerical solution of large-scale nonlinear optimization problems, to solve problem~$\trmp$. In brief, SQP is an iterative algorithm which models the problem for a given iterate $(\ba^{\ell},\bp^{\ell})$, $\ell\in \mathbb{N}_0$, by a Quadratic Programming (QP) subproblem, solves that QP subproblem, and then uses the solution to construct a new iterate $(\ba^{\ell+1},\bp^{\ell+1})$. This construction is implemented in a way that the sequence $(\ba^{\ell},\bp^{\ell})_{\ell\in \mathbb{N}_0}$ converges to a local minimum $(\ba^{*},\bp^{*})$ of the problem as $\ell$ approaches infinity. Note that a major advantage of SQP is that the iterates $(\ba^{\ell},\bp^{\ell})$ need not to be feasible points, since the computation of feasible points in case of nonlinear nonconvex constraint functions may be as difficult as the optimal solution of the problem itself.
	
	For ease of presentation, let us transform the problem given in (\ref{prob:main}) into an equivalent minimization problem, in which we aim to minimize the function
	\[
	\cF(\ba,\bp) = - {\sum}_{q=1}^Q a_{q} - {\sum}_{q=1}^Q R^{\rm{C}}_q(\text{\bf p}),
	\]
	where $(\ba,\bp)\in \mathbb{R}_+^{2Q+K+1}$ is a vector of nonnegative variables, subject to the same set of constraints as in (\ref{prob:main}). Before going to define the Lagrangian function for the problem, we replace the constraint (\ref{eq:min-R-q-0}) by the following constraints.
	\begin{align*}
		\begin{cases} 
			\vspace{4pt}
			\sum_{q=1}^Q a_q\le \Omega\\
			\vspace{4pt}
			\Omega\le R^{\rm{C}}_{q,0}(\bp), \quad \forall~ q\in\cQ.
		\end{cases} 
	\end{align*}
	The Lagrangian function $L(\ba,\bp)$ is then defined as
	\begin{multline*}
		L(\ba,\bp)=	\cF(\ba,\bp) + 	\theta \left(\sum_{q=1}^Qa_{q} - \Omega\right)	+ \sum_{q=1}^Q \sigma_q\left( \Omega -  R^{\rm{C}}_{q,0}(\bp)\right) 
		\\
		+ \sum_{q\in\cQ} \lambda_q\left( C^{\rm{TH}}_q -a_{q}- R^{\rm{C}}_{q}(\bp) \right) + \mu \left({\bar{p}}^{\rm{C}} -p_0- \sum_{q=1}^Q p_q \right)\\
		+ \pmb\xi^T(\pmb\gamma-\vartheta(\bp))+  \pmb\zeta^T(\bp^{\rm{R}}-\bar\bp^{\rm{R}})
		- \pmb\kappa^T\ba - \pmb\eta^T\bp, 
	\end{multline*}
	where $\pmb\gamma=(\gamma^{\rm{R}},\ldots,\gamma^{\rm{R}})$, $\bar\bp^{\rm{R}}=(\bar p^{\rm{R}},\ldots,\bar p^{\rm{R}})$ and $\vartheta(\bp)=(\vartheta_1(\bp),\ldots,\vartheta_K(\bp))$ are vectors of dimension $K$. In addition, $\mu,\theta,\lambda_q,\sigma_q,\pmb\xi, \pmb\zeta,\pmb \kappa,\pmb\eta$ for $q\in\cQ,$ are Lagrange multipliers.
	
	The QP subproblem, denoted as QP$^{(\ell)}$, we will solve in iteration step $\ell$ is defined in such a way  that it should
	reflect the local properties of the problem with respect to the current iterate $(\ba^{\ell},\bp^{\ell})$. In fact,  we replace the objective function $\cF$ by its local quadratic approximation:
	\[
	\cF(\bs)= \cF(\ba^{\ell},\bp^{\ell}) + \nabla \cF(\ba^{\ell},\bp^{\ell})^T \bs +\frac{1}{2} \bs^T \nabla^2L(\ba^{\ell},\bp^{\ell}) \bs,
	\]
	and the nonlinear constraint in (\ref{prob:main}) by their local affine approximations
	\begin{align*}
		g_q(\ba^{\ell},\bp^{\ell}) +\nabla g_q(\ba^{\ell},\bp^{\ell}) ^T \bs \le &\, 0,\\
		h_q(\ba^{\ell},\bp^{\ell}) +\nabla g_q(\ba^{\ell},\bp^{\ell}) ^T \bs \le &\, 0,
	\end{align*}
	where for all $q\in\cQ$, 
	\begin{align*}
		g_q(\ba,\bp) = &\, {\sum}_{q=1}^Qa_{q} - R^{\rm{C}}_{q,0}(\bp) \le 0,\\
		h_q(\ba,\bp) = &\, C^{\rm{TH}}_q - a_{q} - R^{\rm{C}}_{q}(\bp)\le  0.
	\end{align*}
	Note that all the linear constraints in (\ref{prob:main}) remain unchanged. Our goal is to find a feasible solution $\bs\in\mathbb{R}_+^{2Q+K+1}$ that minimizes $\cF(\bs)$.
	
	The quadratic objective function $\cF(\bs)$ of the QP subproblem is nonconvex in general. Hence, the idea is to approximate the Hessian $\nabla^2L(\ba^{\ell},\bp^{\ell})$ 
	by a positive definite matrix, and this can be done using the standard Broyden-Fletcher-Goldfarb-Shanno (BFGS) approximation 
	\cite{head1985broyden}. The resulting objective function is therefore convex and thus the QP subproblem can be solved efficiently to attain the search direction $(\ba^{\ell+1},\bp^{\ell+1})=(\ba^{\ell},\bp^{\ell}) + \alpha_\ell\bs^\ell$, where the step length parameter $\alpha_\ell$ can be evaluated by using, e.g line search algorithms, to perform the following one-dimensional optimization.
	\begin{equation}\label{prob:direct-search}
		\alpha_\ell = \arg\min_{\alpha>0} \cF((\ba^{\ell},\bp^{\ell}) + \alpha \bs^\ell). 
	\end{equation}
	Generally, a line search algorithm aims at solving the problem locally, based on an iterative approach that seeks for a local minimum solution (for example, gradient descent or Newton method). Herein, we make use of the fractional programming approach presented in \cite{ShenY18}, where it is showed that, in practice, this approach has lower complexity than gradient descent and Newton methods on
	per-iteration basis. 
	
	Since $\bs^\ell\in\mathbb{R}_+^{2Q+K+1}$, one can write it as  
	\[
	\bs^\ell= (s_1^\ell,\ldots,s_Q^\ell,s_0^\ell,s_{Q+1}^\ell,\ldots,s_{2Q}^\ell,s_{2Q+1}^\ell,\ldots,s_{2Q+K}^\ell).
	\]
	Also, one can write $(\ba^{\ell},\bp^{\ell})$ as
	\[
	(\ba^{\ell},\bp^{\ell})= (a_1^\ell,\ldots,a_Q^\ell, p_0^\ell,p_1^\ell,\ldots,p_Q^\ell,p^{\ell,\rm{R}}_1,\ldots,p^{\ell,\rm{R}}_K).
	\]
	The function $\cF((\ba^{\ell},\bp^{\ell}) + \alpha \bs^\ell)$ is now written as a function $\cG$ of single variable $\alpha$.
	\begin{align*}
		\cG(\alpha)=  -{\sum}_{q=1}^Q \left(a_{q}^\ell+\alpha s_q^\ell + B\log  \left( \dfrac{f_q(\alpha)}{g_q(\alpha)} \right)\right),
	\end{align*}
	where $f_q(\alpha)$ and $g_q(\alpha)$ are respectively defined as
	\begin{align*}
		h^{\rm{C}}_q\sum_{q'=1}^Q (p_{q'}^\ell+\alpha s_{Q+q'}^\ell)  +   \sum_{k=1}^{K} g^{\rm{RC}}_{k,q}(p^{\ell,\rm{R}}_k+\alpha s_{2Q+k}^\ell) + \sigma^2_q,
	\end{align*}
	and  
	\begin{align*}
		h^{\rm{C}}_q\sum_{{q' = 1, q'\neq q}}^Q (p_{q'}^\ell+\alpha s_{Q+q'}^\ell) +   \sum_{k=1}^{K} g^{\rm{RC}}_{k,q}(p^{\ell,\rm{R}}_k+\alpha s_{2Q+k}^\ell) + \sigma^2_q,
	\end{align*}
	which are linear functions of $\alpha$. We further write the functions $f_q(\alpha)$ and $g_q(\alpha)$ respectively as
	\begin{align*}
		f_q(\alpha)=	\alpha \cdot V_q + W_q, \quad \text{and}\quad g_q(\alpha)=\alpha \cdot V'_q + W'_q,
	\end{align*}
	where 
	\begin{align*}
		V_q=~& h^{\rm{C}}_q{\sum}_{q'=1}^Q s_{Q+q'}^\ell + {\sum}_{k=1}^{K} g^{\rm{RC}}_{k,q} s_{2Q+k}^\ell+ \sigma^2_q\\
		W_q=~& h^{\rm{C}}_q{\sum}_{q'=1}^Q p_{q'}^\ell +  {\sum}_{k=1}^{K} g^{\rm{RC}}_{k,q}p^{\ell,\rm{R}}_k + \sigma^2_q\\
		V'_q=~& h^{\rm{C}}_q{\sum}_{{q' = 1, q'\neq q}}^Q s_{Q+q'}^\ell + {\sum}_{k=1}^{K} g^{\rm{RC}}_{k,q} s_{2Q+k}^\ell+ \sigma^2_q\\
		W'_q=~& h^{\rm{C}}_q{\sum}_{{q' = 1, q'\neq q}}^Q p_{q'}^\ell + {\sum}_{k=1}^{K} g^{\rm{RC}}_{k,q} s_{2Q+k}^\ell+ \sigma^2_q,
	\end{align*}
	
	Now solving the problem (\ref{prob:direct-search}) is equivalent to solving the following problem
	\begin{align}\label{prob:direct-search1}
		\max \limits_{\alpha>0} \quad \quad \sum_{q=1}^Q \left(a_{q}^\ell+\alpha s_q^\ell + B\log\left( \dfrac{\alpha  V_q + W_q}{\alpha  V'_q + W'_q} \right)\right).
	\end{align}
	The problem (\ref{prob:direct-search1}) can be transformed into an equivalent formulation, which can be then amenable to iterative optimization.
	\begin{theorem}\label{th:equivalent}
		The problem (\ref{prob:direct-search1}) is equivalent to the problem
		{\small
			\begin{align}\label{prob:direct-search2}
				\max_{\alpha>0}	~\sum_{q=1}^Q \left( a_{q}^\ell +\alpha s_q^\ell  +B\log\left( 2\beta_q\sqrt{\alpha  V_q + W_q} -\beta_q^2(\alpha  V'_q + W'_q) \right)\right).
			\end{align}
		}
		over the domain $\{(\alpha,\beta_1,\ldots,\beta_q)\in \mathbb{R}_{>0}\times \mathbb{R}^q\}$.
	\end{theorem}
	Since Theorem~\ref{th:equivalent} directly follows from \cite{ShenY18}, the proof of Theorem~\ref{th:equivalent} is omitted. 
	
	We follow the iterative approach to find an optimal solution $\alpha$. First, it is observed that when $\alpha$ is fixed, the optimal $\beta_1,\ldots,\beta_q$ can be determined in closed form as
	\begin{equation}\label{eq:closed-form}
		\beta_q^*=\dfrac{\sqrt{\alpha  V_q + W_q}}{\alpha  V'_q + W'_q},\quad \forall~q\in\cQ.
	\end{equation} 
	On the other hand,  when all $\beta_1,\ldots,\beta_q$ are fixed, the resulting problem is convex  in variable $\alpha$ and thus it can be solved efficiently through numerical convex optimization (e.g., the gradient decent method). As consequence, the complexity of Algorithm~\ref{fractional-programming} is $O(D/\epsilon)$, where $\epsilon$ is a given additive error for solutions $\alpha$ found at every iteration,  and $D$ is the number of iterations.
	\begin{algorithm}[!htb]
		\caption{Fractional Programming (FP)} \label{fractional-programming}
		\begin{algorithmic}[1]
			\Require an accuracy parameter $\epsilon>0$.
			\Ensure a solution $\alpha$
			\State Initiate a solution $\alpha$ of (\ref{prob:direct-search2})
			\Repeat
			\State Update $\beta_1,\ldots,\beta_q$ by (\ref{eq:closed-form})
			\State Update $\alpha$ by solving the convex problem obtained from (\ref{prob:direct-search2}) by fixing $\beta_1,\ldots,\beta_q$ achieved in the previous step \label{step:gradient}
			\Until{Convergence}
			\State \Return $\alpha$ 
		\end{algorithmic}
	\end{algorithm}
	
	Using a similar argument as in \cite{ShenY18}, one can prove that Algorithm~\ref{fractional-programming} 
	consists of a sequence of convex optimization problems that
	converge to a stationary point of (\ref{prob:direct-search1}) with non-decreasing objective value after every iteration.

	The main steps of the SQP algorithm is outlined in Algorithm~\ref{sqp} below. It it easy to see that the complexity of the algorithm mainly depends on the execution time of the repeat-loop. In particular, we need $O((2Q+K+1)^{7/2}\cdot \cI^2)$ times to solve the convex relaxation of each QP subproblem using primal interior point method \cite{GoldfarbL91}, where $\cI$ denotes the length of the input data. In addition, finding a step length $\alpha$ takes $O(D/\epsilon)$ times, while the updating the approximate Hessian matrix costs $O((2Q+K+1)^2)$. Overall, the complexity of  Algorithm~\ref{sqp} is
	\[
	O((2Q+K+1)^{7/2}\cdot \cI^2 + D/\epsilon).
	\]
	\begin{algorithm}[!htb]
		\caption{Sequential Quadratic Programming (SQP)} \label{sqp}
		\begin{algorithmic}[1]
			\Require an accuracy parameter $\epsilon>0$.
			\Ensure a solution $(\ba^{\ell},\bp^{\ell})$
			\State Initiate a solution $(\ba^{0},\bp^{0})$
			\State Compute a positive definite matrix $\pmb{H}_0$ as an approximation of the Hessian $\nabla^2L(\ba^{0},\bp^{0})$
			\State $\ell\leftarrow 0$
			\Repeat
			\State Solve a QP subproblem QP$^{(\ell)}$ to attain a solution $\bs^\ell$
			\State Find a step length $\alpha_\ell$ using Algorithm~\ref{fractional-programming}
			\[
			\alpha_\ell \leftarrow \arg\min_{\alpha>0} \cF((\ba^{\ell},\bp^{\ell}) + \alpha \bs^\ell) 
			\] 
			\State Update 
			$$(\ba^{\ell+1},\bp^{\ell+1}) \leftarrow (\ba^{\ell},\bp^{\ell}) + \alpha_\ell\bs^\ell$$
			\State Update the approximate Hessian $\pmb{H}_{\ell+1}$
			\begin{align*}
				\bu^\ell \leftarrow &\, (\ba^{\ell+1},\bp^{\ell+1}) - (\ba^{\ell},\bp^{\ell}) \\
				\bv^\ell \leftarrow &\, \nabla \cF (\ba^{\ell+1},\bp^{\ell+1}) - \nabla \cF (\ba^{\ell},\bp^{\ell})\\
				\pmb{H}_{\ell+1} \leftarrow &\,  \pmb{H}_{\ell} +\frac{\bv^\ell (\bv^\ell)^T}{(\bv^\ell)^T\bu^\ell} - \frac{\pmb{H}_{\ell} \bu^\ell (\bu^\ell)^T \pmb{H}_{\ell}^T}{(\bu^\ell)^T \pmb{H}_{\ell} \bu^\ell}
			\end{align*}
			\State $\ell\leftarrow \ell+1$
			\Until{Convergence or $\ell=500$}
			\State \Return $(\ba^{\ell},\bp^{\ell})$ 
		\end{algorithmic}
	\end{algorithm}

	\section{A Global Optimization Algorithm}
	\label{sec:exact_algorithm}
	Like most local optimization algorithms, the SQP algorithm as presented in Section~\ref{sec:iter_algorithm} also has a drawback that it can only locate a local optimum, which may be far from a global optimum. In this section, we present an algorithm that can solve problem $\trmp$ globally. The following theorem summarizes main properties of the algorithm.
	\begin{theorem}
		\label{th:exact-alg}
		For a given $\delta>0$, there is an algorithm that can produce a near optimal solution to the problem $\trmp$ with an additive error of $\delta\cdot Q$, where $Q$ denotes the number of users. The running time of the algorithm is exponential in $Q$.
	\end{theorem}
	We prove Theorem \ref{th:exact-alg} by presenting an additive approximation scheme (AAS), which is formally defined as Algorithm~\ref{ptas}. The main idea of the algorithm is to solve the problem (\ref{prob:main}) approximately by transforming it into a sequence of subproblems, whose solution can be efficiently found.  Before giving a formal description of the algorithm, we first study the case where we know the optimal private data rate of all the users in advance, and then show how the algorithm AAS works without this assumption. Throughout this section we denote by $(\ba^*,\bp^*)$  an optimal solution to the problem (\ref{prob:main}). We also call $R^{\rm{C}}_q(\bp^*)$ the optimal private data rate of user $q$, for each $q\in Q$. 
	\subsection{Known Private Data Rates}
	In this section we present a polynomial-time algorithm to find $(\ba^*,\bp^*)$ given the optimal private data rates.
	\begin{lemma}
		\label{lem:core}
		Suppose that we are given optimal private data rates  $R^{\rm{C}}_1(\text{\bf p}^*),\ldots, R^{\rm{C}}_Q(\text{\bf p}^*)$, but not $\bp^*$. Then, the optimal solution $(\ba^*,\bp^*)$  can be found in polynomial time in the input size.
	\end{lemma}
	\begin{proof}
		Suppose that we know in advance the optimal private data rate  $\zeta_q=R^{\rm{C}}_q(\text{\bf p}^*)$, for every $q\in\cQ$. The problem (\ref{prob:main}) is reduced to the following problem:
		\begin{subequations}\label{prob:main-1}
			\begin{align}
				\max \limits_{\ba, \bp\in\cX} & \quad {\sum}_{q=1}^Q a_{q}  \tag{\ref{prob:main-1}}\\
				{\text{s.t.}} & \quad 	{\sum}_{q=1}^Qa_{q} \leq \min \limits_{q\in\cQ} \left\{R^{\rm{C}}_{q,0}(\bp)\right\},\label{eq:min-R-q-1}\\		
				& \quad a_{q}\ge C^{\rm{TH}}_q - \zeta_{q}, \quad \forall q\in \mathcal{Q},\\
				& \quad R^{\rm{C}}_q(\text{\bf p}) = \zeta_q, \quad \quad \, \, \, \forall q\in \mathcal{Q}.
			\end{align}
		\end{subequations}
		
		We prove that the problem (\ref{prob:main-1}) can be solved in polynomial time. It is not hard to check, by contradiction, that the common message constraint (\ref{eq:min-R-q-1}) holds with equality at an optimal solution of (\ref{prob:main-1}). Hence, solving (\ref{prob:main-1}) amounts to  find first an optimal solution $\hat\bp$ to the following max-min problem
		\begin{subequations}\label{prob:main-2}
			\begin{align}
				\max \limits_{\bp\in\cX} & \quad  \min \limits_{q\in\cQ} \left\{R^{\rm{C}}_{q,0}(\bp)\right\}  \tag{\ref{prob:main-2}}\\
				{\text{s.t.}} 
				& \quad R^{\rm{C}}_q(\text{\bf p}) = \zeta_q, \quad \forall q\in \mathcal{Q},
			\end{align}
		\end{subequations}
		and then an optimal solution $\hat\ba$ to the following problem
	\begin{subequations}\label{prob:aq-solve}
		\begin{align}
			\max \limits_{\ba} & \quad {\sum}_{q=1}^Q \hat a_{q}  \tag{\ref{prob:aq-solve}}\\
			{\text{s.t.}} & \quad 	{\sum}_{q=1}^Q\hat a_{q} = \min \limits_{q\in\cQ} \left\{R^{\rm{C}}_{q,0}(\hat\bp)\right\},\\		
			& \quad \hat a_{q}\ge C^{\rm{TH}}_q - \zeta_{q}, \quad \forall q\in \mathcal{Q}.
		\end{align}
	\end{subequations}
	The latter problem is a linear program whose optimal solution $\hat\ba$ can be  determined as
	\begin{subequations} \label{eq:aq-solve}
		\begin{align}
			\nonumber
			{}& \begin{cases} 
				\vspace{4pt}
				\hat a_q=C^{\rm{TH}}_q -  \zeta_q,\quad \forall\, q\in \mathcal{Q}\setminus \{q'\},\\	
				\vspace{4pt}
				\hat a_{q'}=  \min \limits_{q\in\cQ} \{R^{\rm{C}}_{q,0}(\hat\bp)\} -{\sum}_{q=1,q\not=q'}^{Q}\hat a_{q},
			\end{cases} 
		\end{align}
	\end{subequations}
	for any arbitrary choice of user $q'\in\cQ$. 
	It remains to prove that $(\ref{prob:main-2})$ is also polynomially solvable. Indeed, by introducing a new variable $T$, it can be reformulated as
	\begin{subequations}\label{prob:main-3}
		\begin{align}
			\max \limits_{\bp\in\cX} & \quad  T   \tag{\ref{prob:main-3}}\\
			{\text{s.t.}} & \quad 	R^{\rm{C}}_{q,0}(\bp)\ge T,\hspace{0.95cm} \forall q\in \mathcal{Q},\label{eq:min-rate-0-t}\\
			& \quad R^{\rm{C}}_q(\text{\bf p}) = \zeta_q, \hspace{1.01cm} \forall q\in \mathcal{Q} \label{eq:p-rate-q}.
		\end{align}
	\end{subequations}
	Although problem (\ref{prob:main-3}) is nonlinear (due to the nonlinearity of the constraints  (\ref{eq:min-rate-0-t})), it can be solved by using binary search combined with the linear programming method. In fact, the constraints  (\ref{eq:min-rate-0-t}) can be rewritten as
	\[
	h^{\rm{C}}_q{\sum}_{q' = 1}^Q p_{q'} +   {\sum}_{k=1}^{K}  g^{\rm{RC}}_{k,q}p^{\rm{R}}_k +  \sigma^2_q   \le \frac{h^{\rm{C}}_q p_0}{ 2^{T/B}-1}, \quad \forall q\in \mathcal{Q}.
	\]
	For a fixed value of $T$, the above constraints are linear and thus the problem (\ref{prob:main-3}) comes down to checking if a linear program is feasible, which can be done in polynomial time. Consequently, to find an  optimal value $T^*$ of (\ref{prob:main-3}), one can make use of the binary search over the range $[0,\text{\bf T}]$, where $\text{\bf T}$ is some upper bound on the possible maximum value of $T^*$. We will discuss in Section~\ref{sec:upper-bound-data-rate} how to compute such an upper bound  $\text{\bf T}$. The complexity of solving (\ref{prob:main-3}) is $O( (Q+K)^{2.5}\cdot\log(\text{\bf T}))$, since there are totally $Q+K$ variables,  excluding variable $T$.~\end{proof}
\subsection{Unknown Private Data Rates}
We now consider the case where the algorithm does not know the users' optimal private data rates as part of its input. 
Lemma \ref{lem:core} paves the way for dealing with such a case. In fact, instead of using the exact value of the private rates   $R^{\rm{C}}_1(\bp^*),\ldots, R^{\rm{C}}_Q(\bp^*)$, the idea is to work with their approximate value, which can be computed with a multiplicative error of $1+\epsilon$, where $\epsilon = 2^{\delta}-1$. This can be done via an enumeration technique, which we call ``{\em a partition space procedure}", as described below. By the error $1+\epsilon$, it means that the closer the value of $\epsilon$ is to zero, the closer the approximate value of $R^{\rm{C}}_q(\bp^*)$ approaches its exact value. The approximation values can be then used instead of the exact values when solving (\ref{prob:main}) (as shown in the proof of Lemma~\ref{lem:core}), leading to an approximate solution. This solution can be arbitrarily close to the optimal solution $(\ba^*,\bp^*)$ as long as $\epsilon$ (or $\delta$) is sufficiently small. 

In the following, we give a detailed description of the ``{\em  partition space procedure}". The purpose of this procedure is to enumerate a {\em small} number of possible values that could be considered as approximate values of  $R^{\rm{C}}_q(\bp^*)$. To this end, we first need to estimate an upper bound on the value of $R^{\rm{C}}_q(\bp^*)$. Equivalently, by the definition of $R^{\rm{C}}_q(\text{\bf p})$, the estimation can be done by computing an upper bound on the value of
\[
\bar R^{\rm{C}}_q(\text{\bf p})=1 + \frac{h^{\rm{C}}_qp_q}{ h^{\rm{C}}_q{\sum}_{{q' = 1, q'\neq q}}^Q p_{q'} +   {\sum}_{k=1}^{K} g^{\rm{RC}}_{k,q}p^{\rm{R}}_k + \sigma^2_q},
\]
which is denoted as $\ub_q$. We will show later on in Section~\ref{sec:upper-bound-data-rate} how this bound can be obtained.  
We consider the range of the function $\bar R^{\rm{C}}_q(\text{\bf p})$ as $[1,\ub_q]$, where $\ub_q$ is maximum value of $\bar R^{\rm{C}}_q(\text{\bf p})$. Note that the range of the function $R^{\rm{C}}_q(\text{\bf p})$ should be $[0,B\log(\ub_q)]$. 

	\paragraph{Partition Space Procedure}
	For each $q\in \cQ$, we partition the interval $[1,\ub_q]$ into $T_q+1$ consecutive intervals as 
	\begin{equation*}
		\Delta_q= {\bigcup}_{t_q=1}^{T_q+1} \Delta^{t_q}_q, 
	\end{equation*}
	where $T_q=\lfloor \log_{1+\epsilon} \ub_q \rfloor$ and 
	\[
	\Delta^{t_q}_q=[(1+\epsilon)^{t_q-1}, (1+\epsilon)^{t_q}),~ \text{for}~ t_q\in\left\{1, \ldots, T_q\right\},
	\]
	and
	\[
	\Delta^{T_q+1}_q=[(1+\epsilon)^{T_q},\ub_q].
	\]
	The values $(1+\epsilon)^{t_q-1}$ for all $t_q\in\left\{1,\ldots,T_q+1\right\}$ represent approximate values of 
	$\bar R^{\rm{C}}_q(\text{\bf p})$. 
	
	Geometrically, for each $\bp'\in\cX$, the vector $\cR^{\rm{C}}(\bp') =\left[\bar R^{\rm{C}}_1(\bp'),\ldots,\bar R^{\rm{C}}_Q(\bp')\right]$ corresponds to some point lying inside the hypercube 
	$
	\cH=\Delta_1\times\cdots\times \Delta_Q. 
	$
	This hypercube can be seen as a partition into sub-cubes of the form
	\begin{equation}
		\label{box}
		\cH(t_1,\ldots,t_Q)=\Delta^{t_1}_1\times\cdots\times \Delta^{t_Q}_Q\subseteq \cH,
	\end{equation}
	where $(t_1,\ldots,t_Q)\in\left\{1,\ldots,T_1+1\right\}\times \cdots \times \left\{1,\ldots,T_Q+1\right\}$. Note that the number of such sub-cubes is $n=\prod_{q\in\cQ}(T_q+1)\le \prod_{q\in \cQ}(\log_{1+\epsilon} \ub_q +1)$. 
	
	Given $\cH$ and its partition constructed as above, we define a subproblem of  (\ref{prob:main}) as follows. Given a sub-cube $\cH(t_1,\ldots,t_Q)$, (i) check if there any feasible solution $\bp'\in\cX$ for which $\cR^{\rm{C}}(\bp')\in \cH(t_1,\ldots,t_Q)$; and (ii) if yes, return one, among such solutions, along with a vector $\ba'$ such that  the sum of common rates $\sum_{q\in\cQ}a_q$ is maximized subject to the constraint (\ref{eq:min-rate-user}).  Formally, the subproblem can be formulated as
	\begin{subequations}\label{prob:main-4}
		\begin{align}
			\max \limits_{\ba,\bp \in \cX } & \quad   {\sum}_{q\in\cQ}a_q \tag{\ref{prob:main-4}}\\
			{\text{s.t.}} & \quad 	{\sum}_{q\in\cQ}a_q\le \min \limits_{q\in\cQ} \left\{R^{\rm{C}}_{q,0}(\bp)\right\},\label{eq:min-r-qc}\\
			& \quad a_{q}+R^{\rm{C}}_q(\text{\bf p})\ge C^{\rm{TH}}_q , \quad \forall q\in \mathcal{Q},\label{eq:min-rate-q-1}\\
			& \quad  \bar R^{\rm{C}}_q(\text{\bf p})\in \Delta_q^{t_q}, \qquad \quad \, \,\,\, \forall q\in \mathcal{Q}.\label{eq:in-delta-q}		\end{align}
	\end{subequations}
	While the constraints (\ref{eq:in-delta-q}) can be expressed in linear forms
	\begin{equation}\label{eq:linear-form}
		(t_q-1)\log(1+\epsilon)\le \bar R^{\rm{C}}_q(\text{\bf p}) \le t_q\log(1+\epsilon),
	\end{equation}
	the non-convexity of the constraints (\ref{eq:min-r-qc}) and (\ref{eq:min-rate-q-1})  makes problem (\ref{prob:main-4}) difficult to solve exactly. Taking the constraints (\ref{eq:in-delta-q}) into account, it is natural to approximate (\ref{eq:min-rate-q-1}) by the linear constraints
	\begin{equation}\label{eq:linearize}
		a_{q}+ t_q\log(1+\epsilon) \ge C^{\rm{TH}}_q , \quad \forall q\in \mathcal{Q},
	\end{equation}
	with the price that the total rate of each user may decrease by an additive error of $\delta=\log(1+\epsilon)$, as shown in the lemma below. We consider the problem obtained from (\ref{prob:main-4}) by replacing (\ref{eq:min-rate-q-1}) by (\ref{eq:linearize}) as its approximate problem. 
	\begin{lemma}\label{lem:obj-val}
		If (\ref{prob:main-4}) is feasible, its approximated problem is also feasible  and has an optimal solution whose value is within an additive error of $\delta\cdot Q$ from the optimum of (\ref{prob:main-4}).
	\end{lemma}	
	\begin{proof}
		Suppose that  (\ref{prob:main-4}) has an optimal solution $(\overline\ba,\overline\bp)$. By the feasibility of this solution to the constraints  (\ref{eq:min-rate-q-1}), it holds that
		\[
		C^{\rm{TH}}_q\le {\sum}_{q\in\cQ} \left(\overline a_q + R^{\rm{C}}_q (\overline\bp)\right) \le {\sum}_{q\in\cQ} \left(\overline a_{q}+ t_q\log(1+\epsilon) \right),
		\]
		where the second inequality is due to (\ref{eq:linear-form}). This implies the feasibility of the approximated problem of (\ref{prob:main-4}), and let $(\ba',\bp')$ be its optimal solution. Then, we must have that $\sum_{q\in\cQ} a'_q \ge \sum_{q\in\cQ} \overline a_q$, and 
		\begin{align*}
			{\sum}_{q\in\cQ} R^{\rm{C}}_q (\bp')\ge &\, {\sum}_{q\in\cQ} \log(1+\epsilon)^{t_q-1} \\
			= &\, {\sum}_{q\in\cQ} \left(\log(1+\epsilon)^{t_q} - \log(1+\epsilon)\right) \\
			\ge &\, {\sum}_{q\in\cQ} R^{\rm{C}}_q (\overline\bp) -\delta\cdot Q.
		\end{align*}
		Therefore,
		\[
		{\sum}_{q\in\cQ} \left(a'_q + R^{\rm{C}}_q (\bp')\right)\ge {\sum}_{q\in\cQ} \left(\overline a_q + R^{\rm{C}}_q (\overline\bp)\right)-\delta\cdot Q,
		\]
		and this completes the proof of the lemma.~\end{proof}
	Lemma~\ref{lem:obj-val} hints that an optimal solution to the approximated problem of (\ref{prob:main-4}) can be seen as a good approximated solution to (\ref{prob:main-4}). 
	By using the similar proof of Lemma~\ref{lem:core}, one can prove that the approximated problem can be solved in polynomial time.
	
	\begin{algorithm}[!htb]
		\caption{Additive Approximation Scheme (AAS)} \label{ptas}
		\begin{algorithmic}[1]
			\Require a parameter $\delta \in (0,1)$.
			\Ensure a solution with an additive error of $\delta\cdot Q$
			\State $\cP\leftarrow \emptyset$.
			\State $\cR^{\rm{C}}(\text{\bf p})\leftarrow (\bar R^{\rm{C}}_1(\text{\bf p}),\ldots,\bar R^{\rm{C}}_Q(\text{\bf p}))$
			\State Compute upper bounds on the optimal private data rate of users
			\State Define the range of $\cR^{\rm{C}}(\text{\bf p})$ as a hypercube $\cH$
			\State Partition $\cH$ into sub-cubes $\cH(t_1,\ldots,t_Q)$
			\For{each sub-cube $\cH(t_1,\ldots,t_Q)$}\label{step:for-loop}
			\State Find (if any) an optimal solution $(\ba',\bp')$ to the approximated problem of (\ref{prob:main-4})
			\State $\cP\leftarrow \cP\cup (\ba',\bp')$
			\EndFor
			\State $(\tilde\ba,\tilde\bp) \leftarrow \arg\max_{(\ba,\bp)\in\cP} \sum_{q=1}^Q (a_q + R^{\rm{C}}_q(\text{\bf p}))$
			\State \Return $(\tilde\ba,\tilde\bp)$ 
		\end{algorithmic}
	\end{algorithm}
	
	\paragraph{Correctness of Algorithm \ref{ptas}}
	We are now ready to prove Theorem \ref{th:exact-alg}. Suppose that the problem (\ref{prob:main}) is feasible and has an optimal solution $(\ba^*,\bp^*)$. Let $(\tilde\ba,\tilde\bp)$ be the solution returned by Algorithm~\ref{th:exact-alg}. It holds that $\cR^{\rm{C}}(\bp^*)$ must belong to the hypercube $\cH$. In particular, there must be $(t_1,\ldots,t_Q)\in\left\{1,\ldots,T_1+1\right\}\times \cdots \times \left\{1,\ldots,T_Q+1\right\}$ such that
	\[
	\cR^{\rm{C}}(\bp^*)\in\cH (t_1,\ldots,t_Q),
	\]
	or, equivalently,
	\begin{equation}
		\label{eq:epsilon}
		\bar R^{\rm{C}}_q(\bp^*)\in\Delta_q^{t_q},\quad \text{for all}\,\, q\in\cQ.
	\end{equation}
	By Lemma \ref{lem:obj-val}, the approximated problem (\ref{prob:main-4}) w.r.t $\cH (t_1,\ldots,t_Q)$ is feasible as it has $(\ba^*,\bp^*)$ as a feasible solution. Moreover, 
	\[
	{\sum}_{q\in\cQ} \left(a'_q + R^{\rm{C}}_q (\bp')\right)\ge {\sum}_{q\in\cQ} \left(a^*_q + R^{\rm{C}}_q (\bp^*)\right)-\delta\cdot Q,
	\]
	where $(\ba',\bp')$ is an optimal solution returned by Step 6 of Algorithm~\ref{ptas}. By definition of $(\tilde\ba,\tilde\bp)$, we must have that
	\begin{align*}
		{\sum}_{q\in\cQ} \left(\tilde a_q + R^{\rm{C}}_q (\tilde\bp)\right)\ge &\, {\sum}_{q\in\cQ} \left(a'_q + R^{\rm{C}}_q (\bp')\right)\\
		\ge &\, {\sum}_{q\in\cQ} \left( a^*_q + R^{\rm{C}}_q (\bp^*)\right)-\delta\cdot Q.
	\end{align*}
	
	\paragraph{Complexity of Algorithm~\ref{ptas}} For the complexity of Algorithm~\ref{ptas}, it is dominated by the for-loop in step~\ref{step:for-loop}, where we need to solve the relaxed version of a subproblem within a sub-cube $\cH(t_1,\ldots,t_Q)$. Note that the number of sub-cubes is $\prod_{q\in \cQ}(\log_{1+\epsilon} \ub_q +1)$, and the time needed for solving the relaxed problem, which is a linear program with $2Q+K+1$ variables, is $O((2Q+K+1)^{2.5}\cdot\cI)$ (see, e.g \cite{LeeS15}), where $\cI$ denotes the input size. In addition, we will show later that the amount of time needed for computing $\ub_q$ for $q\in\cQ$ is $O(Q(K+1)^{2.5}\cdot\cI)$. Hence, the complexity of the algorithm is 
	\[
	O((2Q+K+1)^{2.5}\cdot \log_{1+\epsilon}^Q\ub + Q(K+1)^{2.5}\cdot\cI),
	\]  
	where $\ub=\max_{q\in\cQ}\left\{\ub_q\right\} $.

	\subsection{Bounding the optimal common and private data rates}\label{sec:upper-bound-data-rate}
	This section is devoted to finding upper bounds on the value of $R^{\rm{C}}_q(\bp^*)$ and of $R^{\rm{C}}_{q,0}(\bp^*)$. We consider the case of $R^{\rm{C}}_q(\bp^*)$ only, since the case of $R^{\rm{C}}_{q,0}(\bp^*)$ can be treated similarly. As argued earlier, it suffice to upper bound the value of $\bar R^{\rm{C}}_q(\bp^*)$. For each $q\in\cQ$, let $\ub_q = 1+\xi_q^{\rm C}$, where $\xi_q^{\rm C}$ is an optimal value to the following single-ratio fractional program
	\begin{subequations}\label{prob:bound}
		\begin{align}
			\max \limits_{\bp } & \quad \frac{h^{\rm{C}}_qp_q}{ h^{\rm{C}}_q{\sum}_{{q' = 1, q'\neq q}}^Q p_{q'} +   {\sum}_{k=1}^{K} g^{\rm{RC}}_{k,q}p^{\rm{R}}_k + \sigma^2_q} \tag{\ref{prob:bound}}\\
			{\text{s.t.}} & \quad \text{(\ref{eq:total-power}), (\ref{eq:SINR}), (\ref{eq:radar-power})},\\
			&\quad p_0,p_q\ge 0, \hspace{2.95cm} \forall q\in \mathcal{Q},
		\end{align}
	\end{subequations}
	The problem (\ref{prob:bound}) is transformable into a linear program, using a variable transformation (as known as the Charnes-Cooper transformation \cite{Frenk2005}) as follows.
	
	Let $y_0=h^{\rm{C}}_qp_q,$ and 
	\begin{equation}\label{eq:variable-trans}
		\begin{aligned}
			t=~&\dfrac{1}{h^{\rm{C}}_q{\sum}_{{q' = 1, q'\neq q}}^Q p_{q'} +   {\sum}_{k=1}^{K} g^{\rm{RC}}_{k,q}p^{\rm{R}}_k + \sigma^2_q},\\
			y_k=~&\dfrac{p^{\rm{R}}_k}{{\sum}_{k'=1}^{K} \tilde g^{\rm{RC}}_{k',q}p^{\rm{R}}_{k'} +  \tilde\sigma_k}, 
		\end{aligned}
	\end{equation}
	for all $k\in\cK$.  Then, (\ref{prob:bound}) is reduced to
	\begin{equation}\label{prob:11}
		\begin{aligned}
			\max \limits_{t,\by} & \quad  y_0 \\
			{\text{s.t.}} & \quad y_0 -  \frac{y_k}{\tilde h_k \gamma^{R}}  +\sum\limits_{\substack{k'=1\\k'\neq k}}^{K} \frac{\tilde{g}_{k',k}}{\tilde h_k}y_{k'}    \le -  \frac{\tilde\sigma_k}{\tilde h_k} t,\,\forall\,k\in\cK, \\
			& \quad {\sum}_{k=1}^{K} \tilde g^{\rm{RC}}_{k,q}y_k +\sigma_q^2 t=1,\\
			& \quad y_k\ge 0,\quad \forall\,k\in\cK\cup\{0\}.
		\end{aligned}
	\end{equation}
	
	The problem (\ref{prob:11}) is a linear program and thus can be solved in polynomial time. Given an optimal solution $(t^*,\by^*)$ to (\ref{prob:11}), one can easily recover the corresponding optimal solution to (\ref{prob:bound}) by using transformation (\ref{eq:variable-trans}). 
	
	We conclude this section by giving a complexity of $O(Q(K+1)^{2.5}\cdot\cI)$ on the finding upper bounds on the value of $R^{\rm{C}}_q(\bp^*)$ and of $R^{\rm{C}}_{q,0}(\bp^*)$, for all $q\in\cQ$. Indeed, for each $q\in\cQ$, we first need to transform the fractional program (\ref{prob:bound}) into the linear program (\ref{prob:11}), and this step takes $O(K)$ times. We need the same amount of time to convert an optimal solution of the latter problem to an optimal solution of the former one. Finally, solving (\ref{prob:11}), which is a linear program with $K+1$ variables, takes $O((K+1)^{2.5}\cdot\cI)$, where $\cI$ denotes the input size.

	\section{Performance Evaluation}\label{sec:perform-eval}
	In this section, we present experimental results to evaluate the two proposed algorithms, i.e., the AAS and SQP.  We implement these  algorithms using Python on an INTEL Core i7, $2.9$ GHz, with $8$ GB of RAM. We consider an RSMA-based CRC system in which the coordination of the CUs, BS, and radars are shown in Fig.~\ref{fig:system_model_1}. The number of CUs varies in the range $\{2,3,4,5,6\}$, which are randomly distributed in a square of $400$ m $\times$ $400$ m on the ground. The BS is at the center of the square with the location of $[0,0,0]$ m. There are two radars, i.e., Radars 1 and 2, and their locations are $[-1000,0,0]$ m and $[1000,0,0]$ m, respectively. The two radars perform tracking a common target, whose coordination is $[0,0,10000]$ m. Note that the distance between any two entities in the network is measured by using the Euclidean norm, i.e., $d_{1,2}= \sqrt{(x_2-x_1)^2 + (y_2-y_1)^2+ (z_2-z_1)^2}$ for entities $1$ and $2$ located at $\left[x_1, y_1, z_1\right]$ and $\left[x_2, y_2, z_2\right]$, respectively. Other simulation parameters are listed in Table~\ref{table:parameters}. 
	
	\begin{figure}[t]
		\centering
		\includegraphics[width=\linewidth]{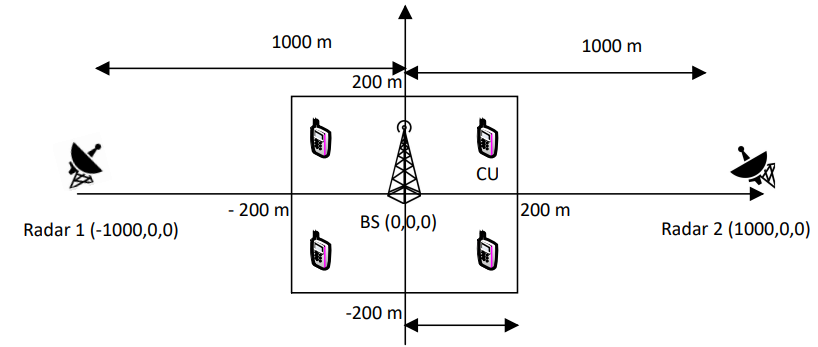}
		\caption{Coordination of RSMA-based CRC system.}
		\label{fig:system_model_1}
	\end{figure}

	\begin{table}[!h]
		\caption{Simulation Parameters}
		\label{table:parameters}
		\centering
		\begin{tabular}{llc}
			\hline\hline
			{Parameters} 		& {\em Value} \\ [0.5ex]
			\hline
			Wave length ($\lambda_c$) & $0.1$ m \\ 
			Maximum power of BS (${\bar{p}}^{\rm{C}}$)   & $30$ dBm \\ 
			Maximum power of radar (${\bar{p}}^{\rm{R}}$)   & $1000$ W  \\ 
			Transmitting antenna gain of BS ($G_{t}^{\text{C}}$)  & $17$ dBi  \\ 
			Receiving antenna gain of CU ($G_{q}$)  & $0$ dBi (1)   \\ 
			Radar antenna gain ($G_{k,t}^{\text{R}}, G_{k,r}^{\text{R}}$)  & $30$ dBi \\ 
			$G_{k,t}^{'\rm{R}}$ & $-27$ dBi~\cite{shi2018non} \\
			$G_{k,r}^{'\rm{R}}$ & $-27$ dBi~\cite{shi2018non}  \\
			$\sigma_k^{\rm{RCS}}, \sigma_{k,k'}^{\rm{RCS}}$   & $1$ m$^2$~\cite{shi2018non} \\ 
			$\sigma^2_q, \sigma^2_k$ &  $-150$ dBm/Hz\\
			$\gamma^R$ &  $10$ dB \\
			Bandwidth $B$ &  $1$ MHz  \\
			$C_q^{\rm{TH}} $ &  $0.1$ Mbs/s \\
			\hline
		\end{tabular}
		\label{table:parameters}
	\end{table}

	\begin{figure}[h!]
		\centering
		\includegraphics[width=0.8\linewidth]{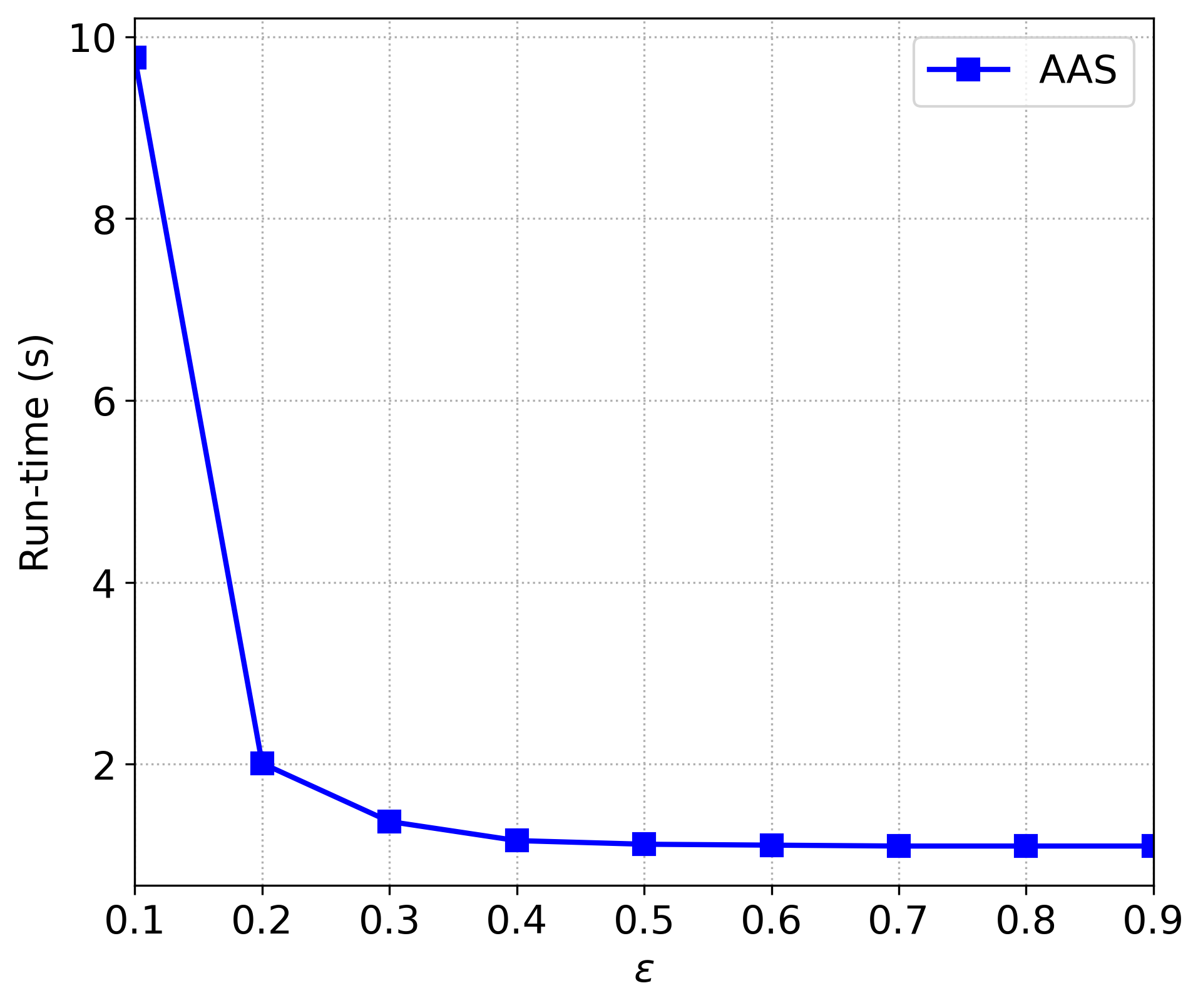}
		\caption{Run-time versus value of $\epsilon$ ($Q=4$ users).}\label{fig:time-epsilon}
	\end{figure}
	
	It is important to properly select the accuracy parameter of $\epsilon$ of the AAS algorithm. The reason can be explained as follows. As presented in Section~\ref{sec:exact_algorithm}, there is a trade-off between the accuracy in the objective value and execution time of AAS. In particular, the proposed algorithm obtains a higher accurate solution as $\epsilon$ is set smaller. However, the algorithm requires more execution time to achieve the accuracy. However, as shown in Fig.~\ref{fig:time-epsilon}, the execution time is low at $\epsilon=0.2$ and keeps constant as $\epsilon \geq 0.4$. In this work, to limit the execution time of AAS algorithm with an acceptable accuracy, we set $\epsilon=0.2$ at which the execution time is around $2$ seconds (see Fig.~\ref{fig:time-epsilon}).
	
	\subsubsection{Execution time of AAS and SQP}
	\label{subsec:executime}
	Figure~\ref{fig:time-user} shows the time that the AAS and SQP algorithms execute in our computing environment. As observed from the figure, given the value of $\epsilon=0.2$ and the number of CUs, the AAS algorithm requires more time than the SQP algorithm. Moreover, as the number of CUs increases, i.e., $Q \geq 5$, the execution time of the AAS algorithm increases rapidly along with the increase of number of CUs. This is due to that the complexity of AAS is an exponential function with respect to the number of CUs as proved in Section~\ref{sec:exact_algorithm}. In contrast, the execution time of SQP slightly changes and always less than half of second.
	\begin{figure}[h!]
		\centering
		\includegraphics[width=0.8\linewidth]{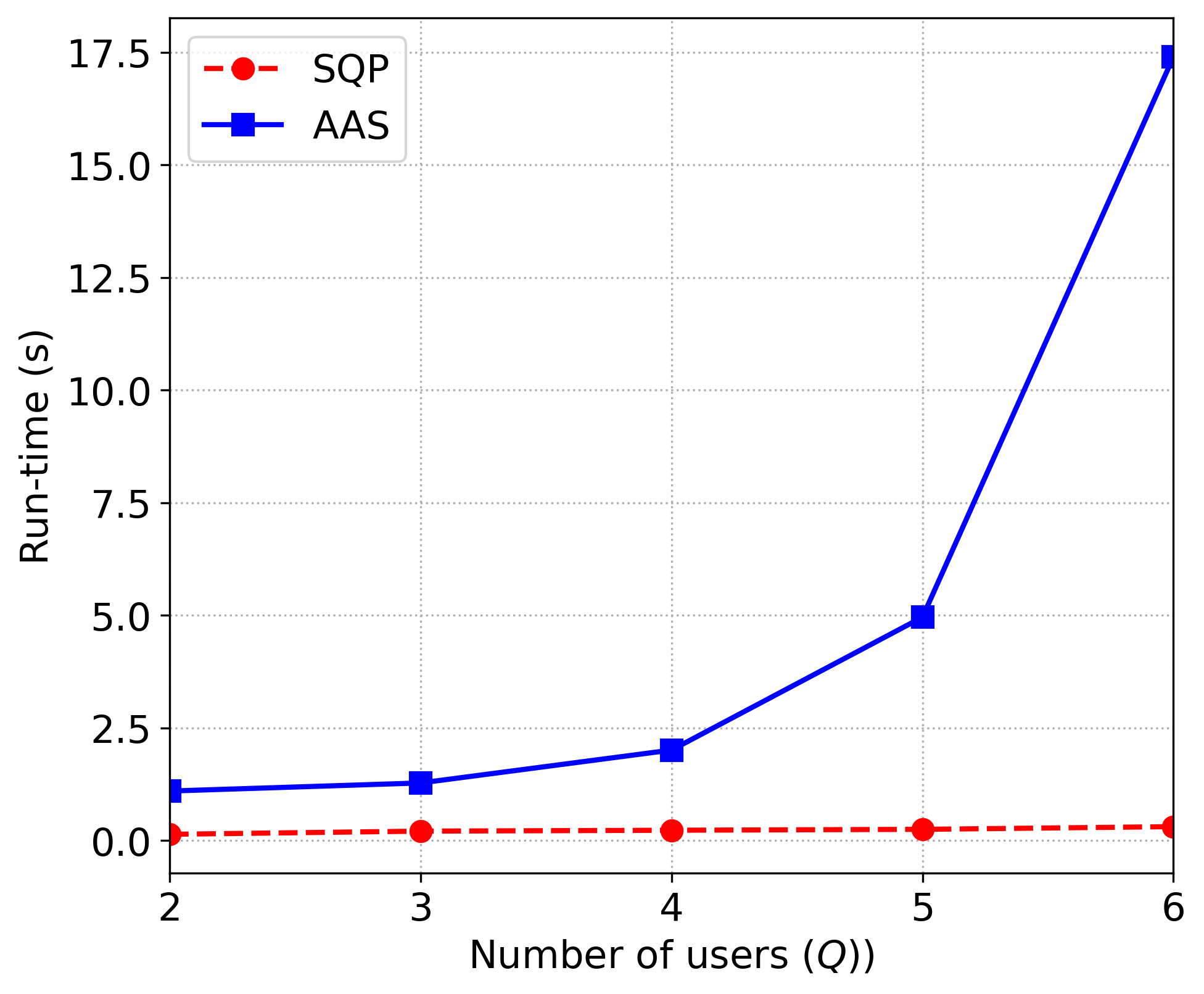}
		\caption{Run-time versus number of users.}\label{fig:time-user}
	\end{figure}

	
	\begin{figure}[h!]
		\centering
		\includegraphics[width=0.8\linewidth]{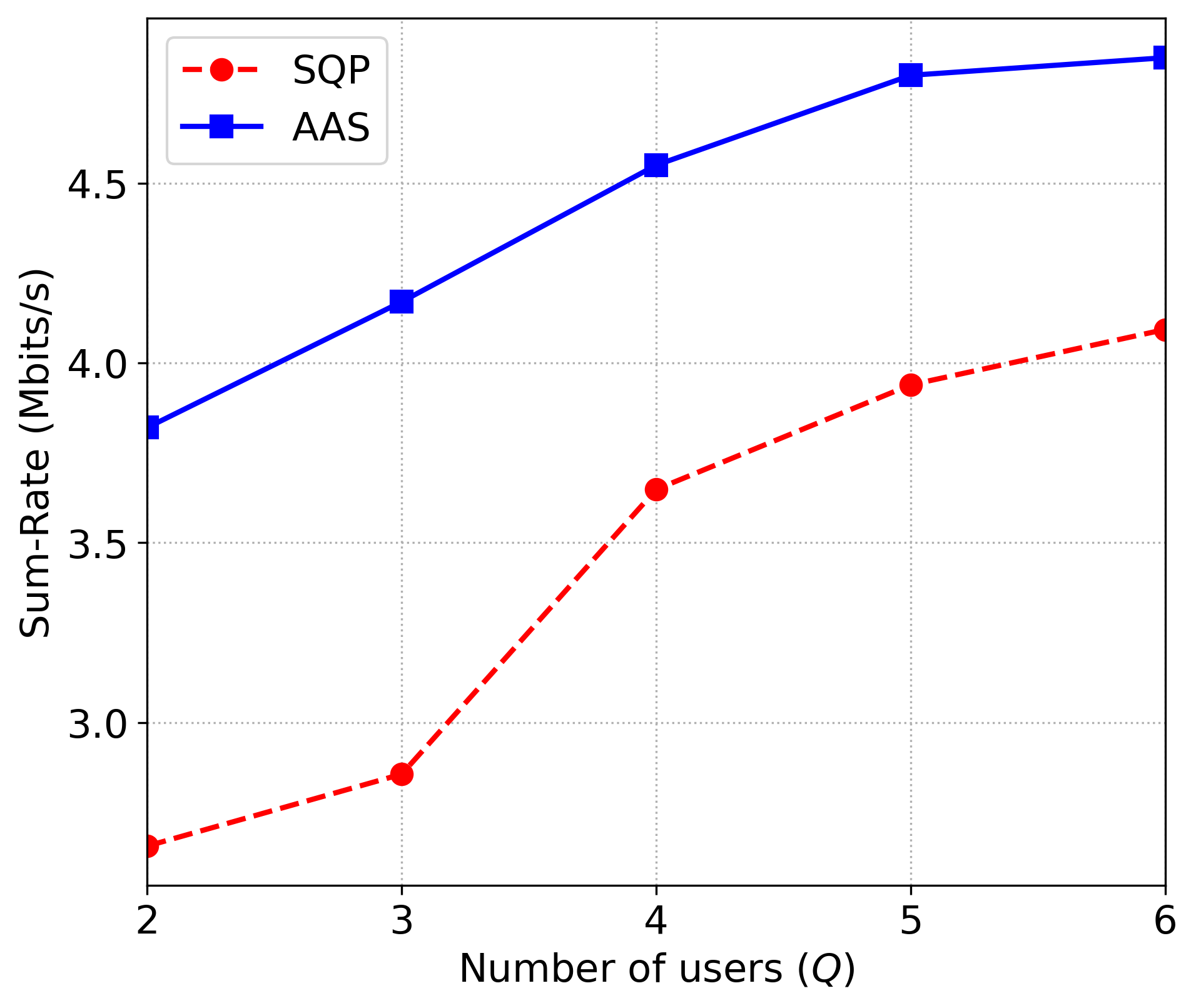}
		\caption{Sum rate versus number of users.}\label{fig:number-user}
	\end{figure}
	\subsubsection{Impact of the number of CUs} Next, we discuss the major objective, i.e., the sum rate of the CUs, obtained by the AAS and SQP algorithms and show how the number of CUs $Q$ affects the sum rate. As shown in Fig. \ref{fig:number-user}, given any the number of CUs, the AAS algorithm always achieves a better performance than SQP. The reason is that the SQP algorithm only produces the suboptimal solution to the $\trmp$ problem, while the AAS algorithm is able to find an almost exact optimal solution, as long as the accuracy $\epsilon$ is chosen to be small enough. Note that $\epsilon$ can be set such that the sum rate obtained by the AAS algorithm is further higher at
	the cost of execution time (as discussed in Section~\ref{subsec:executime}). Figure~\ref{fig:number-user} further shows that the sum rate of the CUs obtained by both the AAS and SQP algorithms rapidly increases as the number of CUs increases. This is due to the multiuser gain of the RSMA scheme~\cite{yang2021optimization}, compared with conventional NOMA.

	\begin{figure}[h!]
		\centering
		\includegraphics[width=0.8\linewidth]{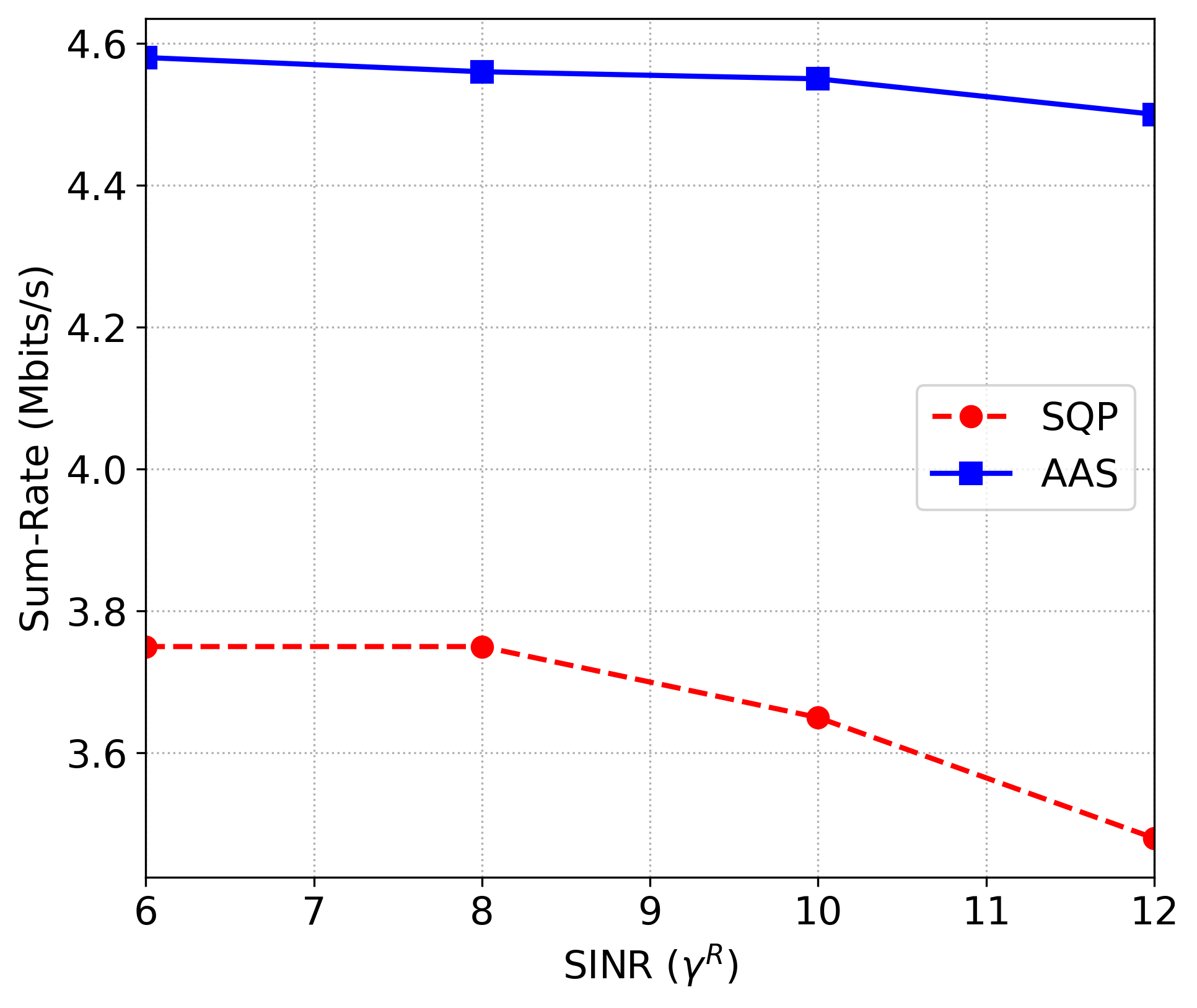}
		\caption{Sum rate versus the SINR threshold $\gamma^R$ ($Q=4$ users).}\label{fig:SINR}
	\end{figure}
	
	\subsubsection{Impact of the SINR requirement of radars} Our work aims to maximize the sum rate of the CUs, subject to the radars' SINR requirements, i.e., $\gamma^R$. Therefore, it is necessary to discuss how the sum rate of the CUs changes as the SINR requirements of the radars change. As illustrated in Fig.~\ref{fig:SINR}, given the number of users of $Q=4$, as $\gamma^R$ increases, the sum rate of the CUs in the RSMA-based communication system decreases. This is explained based on the constraint given in (\ref{eq:SINR}) in which the total transmit power of of the CUs is inversely proportional to $\gamma^R$. Thus, the increase of $\gamma^R$ decreases the sum rate of the CUs. It is worth noting that with the AAS, the sum rate of the CUs slightly decreases with the increase of $\gamma^R$, while with the the SQP algorithm, the sum rate of the CUs  dramatically decreases with the increase of $\gamma^R$. This result shows the effectiveness of the AAS algorithm compared with the SQP algorithm. We can say that with the AAS algorithm, the RSMA-based communication network better coexists with the radars. 
	
	\begin{figure}[h!]
		\centering
		\includegraphics[width=0.8\linewidth]{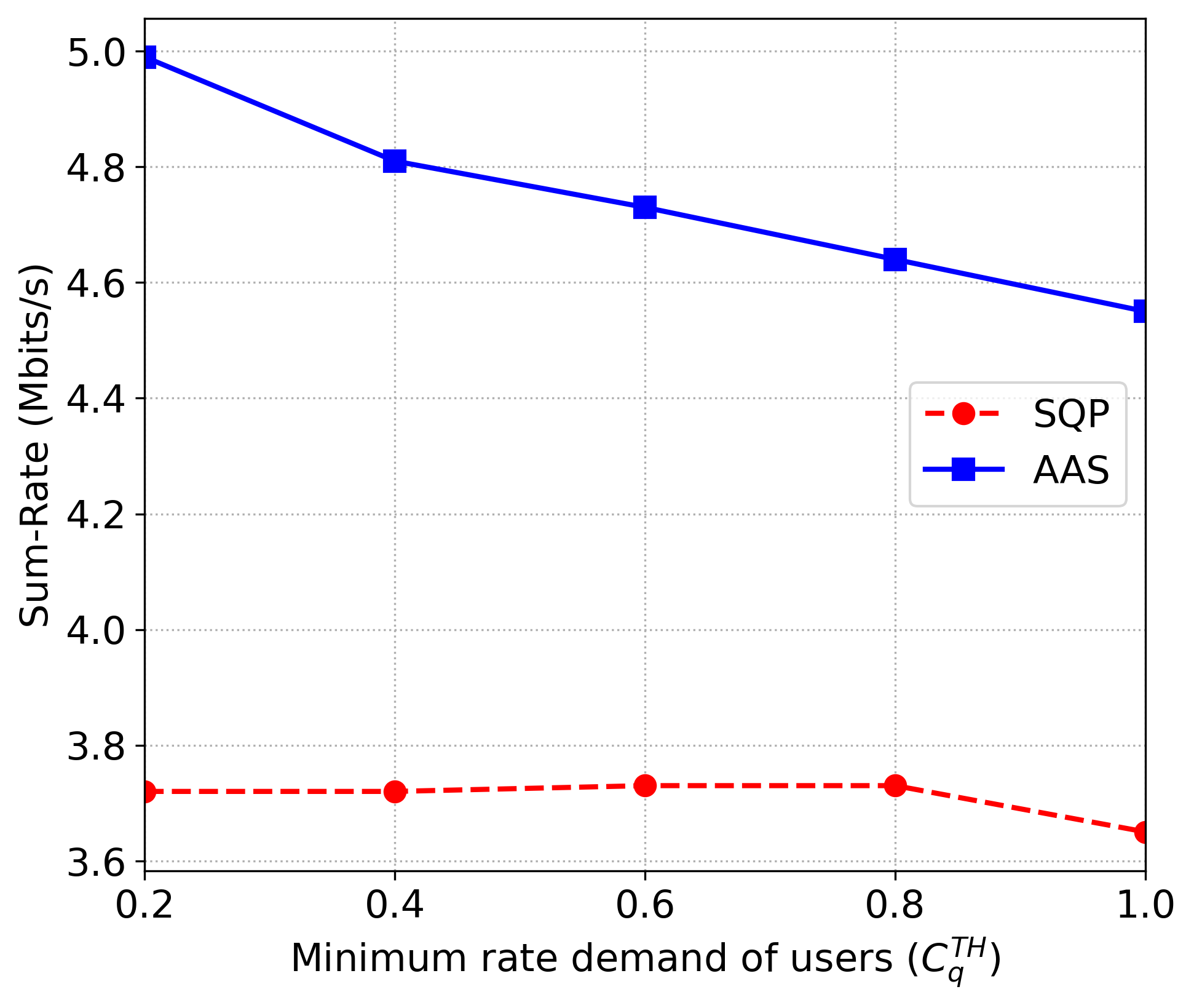}
		\caption{Sum rate versus minimum rate demand of each user ($Q=4$ users).}\label{fig:min-rate}
	\end{figure}
	
	\subsubsection{Impact of rate requirement of CUs} Figure~\ref{fig:min-rate} shows the sum rate of the CUs obtained by the AAS and SQP algorithms versus the minimum rate requirement of the CUs. Here, the minimum rate requirement of the CUs varies in the range of $[0.2;1.0]$ Mbs/s. As seen, the AAS algorithm always outperforms the SQP algorithm at every minimum rate requirement. On the other hand, it is observed from the figure that a high minimum rate requirement results in a low sum rate, especially when the minimum rate requirement goes up from $0.8$ to $1$. This is due to the fact that one may have to increase the transmit power $p_0$ of the common massage to obtain a higher common data rate, leading to a high common rate allocated to users. This, however, consequently degrades the sum rate of all users in the network.		\begin{figure}[h!]
		\centering
		\includegraphics[width=0.8\linewidth]{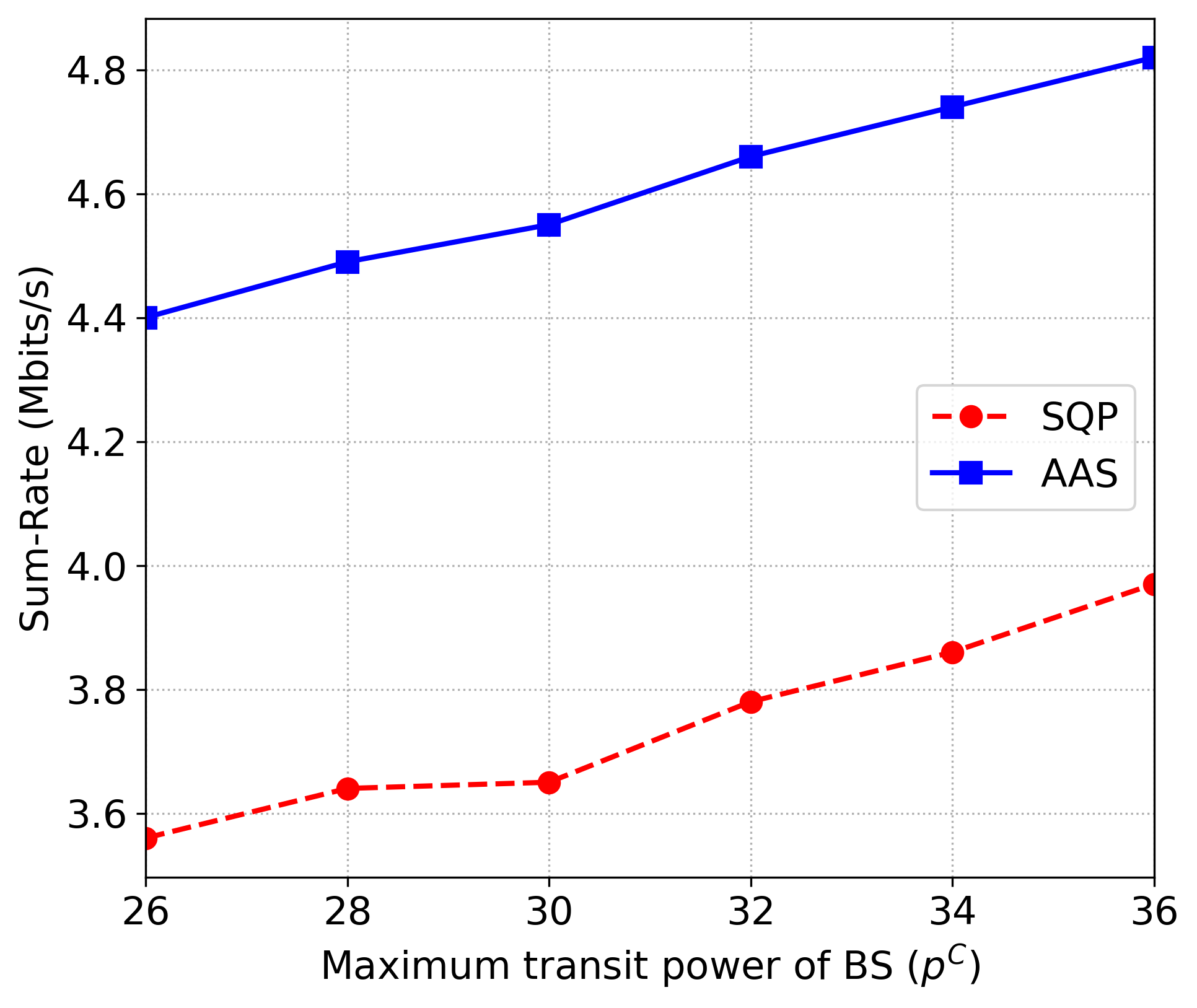}
		\caption{Sum rate versus maximum transmit power of the BS ($Q=4$ users).}\label{fig:pc}
	\end{figure}
	
	\subsubsection{Impact of the power budget of the BS} Figure~\ref{fig:pc} shows the sum rate of the CUs versus the power budget of the BS. As seen, the higher the  power budget is, the higher the sum rate is. This is obvious since the higher power budget allows the BS to allocate more power resources to the CUs, which leads to the higher sum rate.  In addition, it is seen that the AAS algorithm can achieve up to $26\%$ gain in terms of sum rate compared with the SQP algorithm. Again, the reason is that the SQP algorithm only produces
	the suboptimal solution to the TRMP problem, while the AAS
	algorithm is able to find an almost exact optimal solution.

	\begin{figure}[h!]
		\centering
		\includegraphics[width=0.8\linewidth]{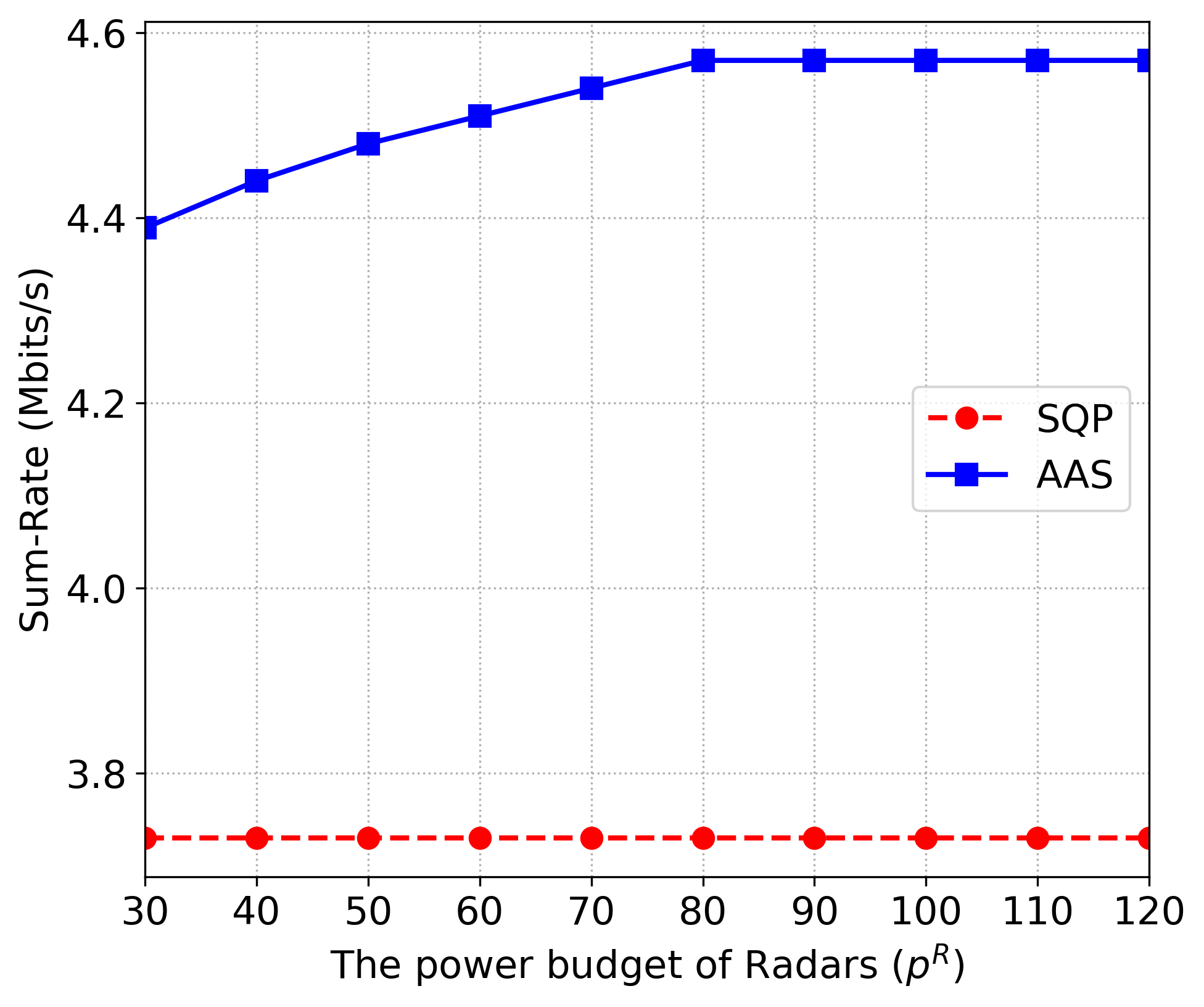}
		\caption{Sum rate versus maximum transmit power of the radars ($Q=4$ users).}\label{fig:pr}
	\end{figure}		
	
	\subsubsection{Impact of the power budget of the radars} Figure~\ref{fig:pr} illustrates the sum rate of the CUs obtained by the AAS and SQP algorithms when the power budget of the radars varies, i.e., in the range $\{30W,\ldots,120W\}$. As seen, the sum rate obtained by the SQP algorithm seems to be unvarying and around $3.7$ Mbits/s, while the sum rate achieved by AAS is much higher and increases from $4.4$ Mbits/s to nearly $4.6$ Mbits/s. Also, it is observed that the sum rate by AAS reaches the highest value at $p^{\rm R}\approx 80W$ and keeps unchanged when the power budget of radars goes beyond. This is because the optimal solution of the problem is attained when the power of every radar is around $80W$. To explain the reason that the sum rate increases as $p^{\rm R}$ increases, we reformulate the radars' SINR constraints in the following form
	\[
	{\sum}_{q = 0}^Q {{p_q}} \le \min_{k\in\cK}\left\{ \frac{p^{\rm{R}}_k}{\gamma^R\tilde h_k} - {\sum}_{{k'\neq k}} \frac{\tilde g_{k',k}}{\tilde h_k}      p^{\rm{R}}_{k'} -\frac{\tilde\sigma_k}{\tilde h_k} \right\}.
	\]
	Theoretically, we would like to find a power allocation of radars, given the range of $p^k\in[0,p^{\rm R}]$, so as to maximize the left-hand side of the equation above, which is a piece-wise linear function in $p_k$, $k\in\cK$. Moreover, the left-hand side is concave over $[0,p^{\rm R}]^K$, and thus has a unique global optimum in this domain. This explains why the sum rate objective, which is directly proportional to $\sum_{q = 0}^Q {{p_q}}$, has the shape of a concave function over the range $[30;120]$. We can explain the results shown in Fig.~\ref{fig:pr} in a different way as follows. As the power budget of the radars increases, the radar transmission power increases. With the effective interference management, the RSMA scheme enables the BS to significantly increase the transmission power of the BS, compared with the increase of the radar transmission power. This results in the increase of the sum rate. However, the increase in sum rate becomes marginal when the radar transmission power reaches $80$W due to the fact that the interference from radar is out of the interference management capability of RSMA. In this case, although the transmission power of BS keeps increasing as the radar transmission power increases, there will be no improvement in the sum rate.
	
	
	\begin{figure}[h!]
		\centering
		\includegraphics[width=0.8\linewidth]{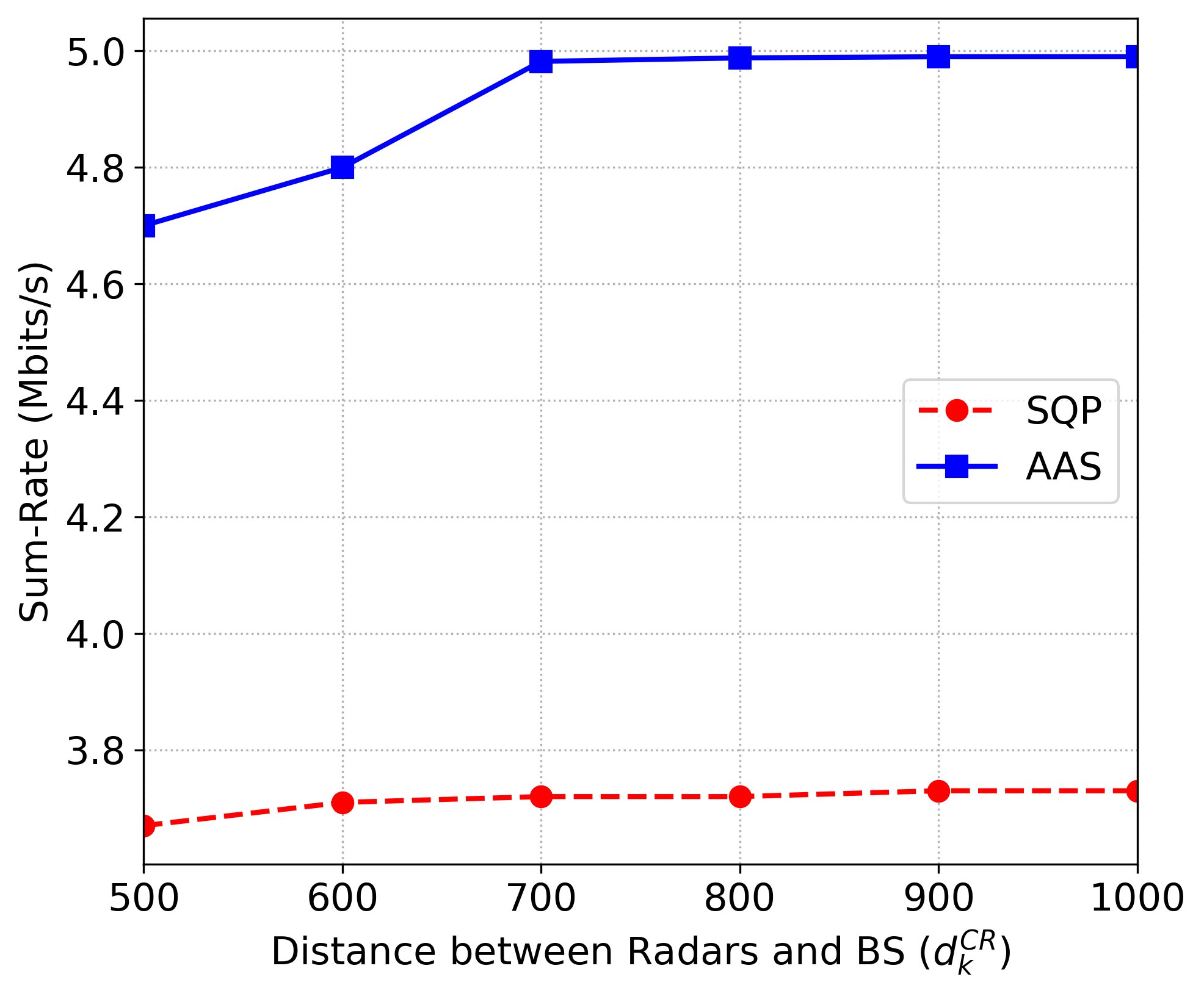}
		\caption{Sum rate versus distance (in km) between Radar 2 and BS ($Q=4$ users).}\label{fig:distance}
	\end{figure}
	\subsubsection{Impact of distance between the radars and BS} Finally, it is interesting to discuss how the distances between the radars and BS impact on the sum rate of the CUs, which helps to find appropriate locations for the radars. Figure~\ref{fig:distance} illustrates the sum rate versus the distance between Radar 2 and BS (we fix the location of Radar 1). For both the algorithms AAS and SQP, the sum rate increases when the distance between Radar 2 and the BS increases. This is because when the radar is located far from the BS and users' positions, the channel gain between the radar and users significantly decreases, leading to a higher data rate for the users. It is also observed from Fig.~\ref{fig:distance} that at the distance larger than $700$m, the sum rate seems to be unchanged. This implies that the radars should be deployed at least $700$m away from the BS so as not to affect the performance of the communication systems.
	
	\section{Conclusion}
	\label{conclusion}
	
	In this paper, we have investigated the coexistence of the RSMA-based communication network and multiple radars, i.e., the RSMA-based CRC system. In particular, we have formulated an optimization problem that optimize the common rates of the CUs, the transmit power of the common and private messages, and the transmit power of radars. The objective is to maximize the sum rate of all the CUs subject to the requirements of their data rates and those of the radars' SINR. To solve the optimization problem, we proposed two algorithms, i.e., SQP and AAS, for solving the problem locally and globally, respectively. 
	We have provided simulation results to demonstrate the improvement and effectiveness of the proposed algorithms. In particular, it has been shown that, although AAS is theoretically of high complexity, it significantly outperforms SQP in terms of objective value in all the considered cases. Furthermore, for the case of less than $5$ users, AAS can return a solution within a reasonable amount of time. In addition, through the simulation results, the best locations for the radars can be also determined.


	\bibliographystyle{IEEEtran}
	\bibliography{bibfile}{}
	
\end{document}